\documentclass[aap,MSNbibl,seceqn,citesort,noautosecdot,dvips]{arximspdf}

\doi{10.1214/11-AAP799}
\volume{22}
\issue{4}
\pubyear{2012}
\firstpage{1465}
\lastpage{1494}

\makeatletter
\newtheorem{prop}{Proposition}[section]
\newtheorem{lem}{Lemma}[section]

\newproclaim{rem}{Remark}[section]
\newproclaim{exm}{Example}[section]
\newproclaim{asm}{Assumption}[section]

\renewcommand{\P}{\mathbb{P}}
\newcommand{\R}{\mathbb{R}}
\newcommand{\E}{\mathbb{E}}
\newcommand{\F}{\mathcal{F}}
\newcommand{\wW}{\widetilde{W}}
\newcommand{\M}{\mathcal{M}}
\newcommand{\eps}{\varepsilon}

\renewcommand{\tilde}{\widetilde}

\makeatother

\begin{document}
\begin{frontmatter}

\title{Outperforming the market portfolio with a~given~probability}
\runtitle{\hspace*{-18pt}Outperforming the market portfolio with a given probability}

\begin{aug}
\author[A]{\fnms{Erhan} \snm{Bayraktar}\corref{}\thanksref{t1}\ead[label=e1]{erhan@umich.edu}},
\author[A]{\fnms{Yu-Jui} \snm{Huang}\ead[label=e2]{jayhuang@umich.edu}}
\and
\author[B]{\fnms{Qingshuo} \snm{Song}\thanksref{t2}\ead[label=e3]{song.qingshuo@cityu.edu.hk}}
\runauthor{E. Bayraktar, Y.-J. Huang and Q. Song}
\affiliation{University of Michigan, University of Michigan and
City~University~of~Hong~Kong}
\address[A]{E. Bayraktar\\
Y.-J. Huang\\
Department of Mathematics\\
University of Michigan\\
Ann Arbor, Michigan 48109\\
USA\\
\printead{e1}\\
\hphantom{E-mail: }\printead*{e2}} 
\address[B]{Q. Song\\
Department of Mathematics\\
City University of Hong Kong\\
China\\
\printead{e3}}
\end{aug}

\thankstext{t1}{Supported in part by the National Science Foundation under
an applied mathematics research Grant DMS-09-06257, a Career Grant
DMS-09-55463 and Susan M. Smith Professorship.}

\thankstext{t2}{Supported in part by City University of Hong Kong
(Project No. 7002677)
and the Research Grants Council of Hong Kong (Project No. CityU 104007).}

\received{\smonth{2} \syear{2011}}
\revised{\smonth{6} \syear{2011}}

%
\begin{abstract}
Our goal is to resolve a problem proposed by Fernholz and Karat\-zas [On
optimal arbitrage (2008) Columbia Univ.]: to characterize the minimum
amount of initial capital with which an investor can beat the market
portfolio with a certain probability, as a function of the market
configuration and time to maturity. We show that this value function is
the smallest nonnegative viscosity supersolution of a nonlinear PDE. As
in Fernholz and Karatzas [On optimal arbitrage (2008) Columbia Univ.],
we do not assume the existence of an equivalent local martingale
measure, but merely the existence of a local martingale deflator.
\end{abstract}

%

\end{frontmatter}

\section{\texorpdfstring{Introduction.}{Introduction}}

In this paper we consider the quantile hedging problem when the
underlying market does not have an equivalent martingale measure.
Instead, we assume that there exists a~\textit{local martingale deflator}
(a~strict local martingale which, when multiplied by the asset prices,
yields a~positive local martingale). We characterize the value function
as the smallest nonnegative viscosity supersolution of a fully
nonlinear partial differential equation. This resolves the open problem
proposed in the final section of~\cite{FK08}; also see pages 61 and 62
of~\cite{Ruf-thesis}.

Our framework falls under the umbrella of the stochastic portfolio
theory of Fernholz and Karatzas (see, e.g.,
\cite{fernholzbook,FKK05,FK08a}) and the benchmark approach of
Platen~\cite{platenbook}. In this framework, the linear partial
differential equation that the superhedging price satisfies does not
have a unique solution; see, for example,~\cite{FK10OA,FK08a,DFKOct09}
and~\cite{Ruf2010}. Similar phenomena occur when the asset prices have
\textit{bubbles}: an equivalent local martingale measure exists, but
the asset prices under this measure are strict local martingales; see,
for example, \cite
{Cox-Hobson,Heston-Loewenstein-Willard,JPS,Ekstrom-Tysk-bubble,BX09}
and~\cite{Jarrow-Protter-Shimbo}. A related series of papers
\cite{Andersen-Piterbarg,Lions-Musiela,Hobson,Lewis,et09,bkx10}
and~\cite{Sin}, addressed the issue of bubbles in the context of
stochastic volatility models. In particular, Bayraktar, Kardaras and
Xing~\cite{bkx10} gave necessary and sufficient conditions for linear
partial differential equations appearing in the context of stochastic
volatility models to have a unique solution.

In contrast, we show that the quantile hedging problem, which is
equivalent to an optimal control problem, is the smallest nonnegative
viscosity supersolution to a fully nonlinear PDE. As in the linear
case, these PDEs may not have a unique solution, and, therefore, an
alternative characterization for the value function needs to be
provided. Recently, the authors of
\cite{bkx09,FK10} and~\cite{Karatzas-Kardaras} also considered
stochastic control problems in
this framework. The first reference solves the classical utility
maximization problem; the second one solves the optimal stopping
problem; whereas the third one determines the optimal arbitrage under
model uncertainty, which is equivalent to solving a zero-sum stochastic game.

The structure of the paper is simple: in Section \ref{eqmodel}, we
formulate the problem. In this section we also discuss the implications
of assuming the existence of a local martingale deflator. In
Section \ref{secquantile-hedging}, we generalize the results of
\cite{FL99} on quantile hedging, in particular the Neyman--Pearson
lemma. We
also prove other properties of the value function such as convexity.
Section \ref{secPDE-characterization} is where we give the PDE
characterization of the value function.

\section{\texorpdfstring{The model.}{The model}}\label{eqmodel}

We consider a financial market with a bond which is always equal to
$1$, and $d$ stocks $X=(X_1,\ldots, X_d)$
which satisfy
%
\begin{eqnarray}
\label{eqstk}
d X_i(t) &=& X_i(t) \Biggl( b_i(X(t)) \,dt + \sum_{k=1}^d s_{ik} (X(t)) \,d
W_k(t)\Biggr),\nonumber\\[-8pt]\\[-8pt]
&&\eqntext{i =1, \ldots, d, X(0)=x=(x_1, \ldots, x_d),}
\end{eqnarray}
where $W(\cdot):=(W_1(\cdot),\ldots,W_d(\cdot))$ is a $d$-dimensional
Brownian motion.

Following the set up in~\cite{FK10OA}, Section 8, we make the following
assumption.

\begin{asm}\label{asfassp}
Let $b_i\dvtx (0,\infty)^d \to\R$ and $s_{ik}\dvtx(0,\infty)^d \to\R$ be
continuous functions. Set $b(\cdot)=(b_1(\cdot), \ldots, b_d(\cdot))'$
and $s(\cdot)=(s_{ij}(\cdot))_{1 \leq i, j \leq d}$, which we assume to
be invertible for all $x \in(0,\infty)^d$. We also assume that (\ref
{eqstk}) has a weak solution that is unique in distribution for every
initial value. Let $(\Omega,\mathcal{F},\P)$\vadjust{\goodbreak} denote the probability
space specified by a weak solution. 
Another assumption we will impose is that
%
\begin{equation}\label{eqass-mrp}
\sum_{i=1}^{d}\int_0^{T}\bigl(|b_i(X(t))|+a_{ii}(X(t))+\theta
_i^{2}(X(t))\bigr)\,dt<\infty,\qquad
\mathbb{P}\mbox{-a.s},
\end{equation}
where $\theta(\cdot):=s^{-1}(\cdot)b(\cdot)$, $a_{ij}(\cdot):=\sum
_{k=1}^d s_{ik}(\cdot)s_{jk}(\cdot)$.
\end{asm}

We will denote by $\mathbb{F}=\{\mathcal{F}_t\}_{t\ge0}$ the
right-continuous version of the natural filtration generated by
$X(\cdot
)$, and by $\mathbb{G}$ the $\mathbb{P}$-augmentation of the
filtration~$\mathbb{F}$. Thanks to Assumption \ref{asfassp}, the Brownian motion
$W(\cdot
)$ of (\ref{eqstk}) is adapted to~$\mathbb{G}$ (see, e.g.,
\cite{FK10OA}, Section 2), every local martingale of $\mathbb{F}$ has the
martingale representation property; that is, it can be represented as a
stochastic integral, with respect to $W (\cdot)$, of some $\mathbb
{G}$-progressively measurable integrand (see, e.g., the discussion on
page~1185 in~\cite{FK10OA}), the solution of (\ref{eqstk}) takes
values in the positive orthant and the exponential local martingale
%
\begin{eqnarray}
\label{eqdfl}
Z(t) := \exp\biggl\{ -\int_0^t \theta(X(s))' \,d W(s) -
\frac1 2 \int_0^t |\theta(X(s))|^2 \,ds \biggr\},\nonumber\\[-8pt]\\[-8pt]
&&\eqntext{0 \leq t<\infty,}
\end{eqnarray}
the so-called \textit{deflator} is well defined. We do not exclude the
possibility that~$Z(\cdot)$ is a strict local martingale.

Let $\mathcal{H}$ be the set of $\mathbb{G}$-progressively measurable
processes \mbox{$\pi\dvtx[0,T)\times\Omega\to\R^d$}, which satisfies
\[
\int_0^{T} \bigl(|\pi(t)'\mu(X(t))|+\pi(t)'\alpha(X(t))\pi
(t)
\bigr)\,dt<\infty, \qquad\mathbb{P}\mbox{-a.s.},
\]
in which $\mu=(\mu_1, \ldots, \mu_d)$ and $\sigma= (\sigma_{ij})_{1
\leq i, j \leq d}$ with
$\mu_i(x)=b_i(x)x_i$, $\sigma_{ik}(x)=s_{ik}(x)x_i$ and $\alpha
(x)=\sigma(x)\sigma(x)'$.

At time $t$, an investor invests $\pi_i(t)$ proportion of his wealth in
the $i$th stock. The proportion $1-\sum_{i=1}^d \pi_i(t)$ gets invested
in the bond.
For each $\pi\in\mathcal{H}$ and initial wealth $y \geq0$ the
associated wealth process will be denoted by $Y^{y,\pi}(\cdot)$. This
process solves
\[
d Y^{y,\pi}(t) = Y^{y,\pi}(t) \sum_{i=1}^d
\pi_i(t)\frac{d X_i(t)}{X_i(t)},\qquad
Y^{y,\pi}(0)=y.
\]

It can be easily seen that $Z(\cdot) Y^{y,\pi}(\cdot)$ is a positive
local martingale for any \mbox{$\pi\in\mathcal{H}$}. Let $g\dvtx(0,\infty)^d
\to
(0,\infty)$ be a measurable function satisfying
%
\begin{equation}\label{eqasmong}
\mathbb{E}[Z(T)g(X(T))]<\infty
\end{equation}
and define
\[
V(T,x,1):=\inf\{y>0\dvtx \exists \pi(\cdot) \in\mathcal{H} \mbox
{ s.t. } Y^{y,\pi}(T) \geq g(X(T))\}.
\]
Thanks to Assumption \ref{asfassp}, we have that $V(T,x,1) =
\mathbb{E}[Z(T)g(X(T))]$; see, for example,~\cite{FK08a}, Section 10.
Note that if $g$ has linear growth, then~(\ref{eqasmong}) is satisfied
since the process $ZX$ is a positive supermartingale.

\subsection{\texorpdfstring{A digression: What does the existence of a local martingale
deflator entail\textup{?}}{A digression: What does the existence of a local martingale
deflator entail}}
Although we do not assume the existence of equivalent local martingale
measures, we assume the existence of a local martingale deflator. This
is equivalent to the \textit{No Unbounded Profit with Bounded Risk}
(NUPBR) condition; see~\cite{Karatzas-Kardaras}, Theorem~4.12.
NUPBR is defined as follows:
a sequence $(\pi^n)$ of admissible portfolios is said to generate a
UPBR if $\lim_{m \to\infty} \sup_{n}\mathbb{P}[Y^{1,\pi^n}(T)>m]>0$.
If no such sequence exists, then we say that NUPBR holds; see
\cite{Karatzas-Kardaras}, Proposition 4.2. In fact, the so-called
\textit{No Free Lunch with Vanishing Risk} (NFLVR) is equivalent to NUPBR plus the
classical \textit{no-arbitrage} assumption. Thus, in our setting (since
we assumed the existence of local martingale deflators) although
arbitrages exist they remain on the level of ``cheap thrills,'' which
was coined by~\cite{lw2000}. (Note that the results of Karatzas and
Kardaras~\cite{Karatzas-Kardaras} also imply that one does not need
NFLVR for the portfolio optimization problem of an individual to be
well defined. One merely needs the NUPBR condition to hold.) The
failure of no-arbitrage means that the money market is not an optimal
investment and is dominated by other investments. It follows that a
short position in the money market and long position in the dominating
assets leads one to arbitrage. However, one cannot scale the arbitrage
and make an arbitrary profit because of the admissibility constraint,
which requires the wealth to be positive. This is what is contained in
NUPBR, which holds in our setting. Also, see~\cite{kFTAP}, where these
issues are further discussed.

\section{\texorpdfstring{On quantile hedging.}{On quantile hedging}}\label{secquantile-hedging}
In this section, we develop new probabilistic tools to extend results
of F{\"o}llmer and Leukert~\cite{FL99} on quantile hedging
to settings where equivalent martingale measures need not exist. This
is not only mathematically intriguing, but also economically important
because it admits arbitrage in the market, which opens the door to the
notion of optimal arbitrage, recently introduced in Fernholz and
Karatzas~\cite{FK10OA}. The tools in this section facilitate the
discussion of quantile hedging under the context of optimal arbitrage,
leading us to generalize the results of~\cite{FK10OA} on this sort of
probability-one outperformance.

We will try to determine
%
\begin{equation}
\label{eqprb}\quad
V(T,x,p) = \inf\bigl\{y>0 | \exists \pi\in\mathcal{H} \mbox{ s.t. }
\mathbb{P} \{Y^{y, \pi}(T) \ge g(X(T))\} \ge p\bigr\}
\end{equation}
for $p \in[0,1]$. Note that the set on which we take infimum in (\ref
{eqprb}) is nonempty. Indeed, under condition (\ref{eqasmong}), there
exists $\pi\in\mathcal{H}$ such that $Y^{y,\pi}(T)=g(X(T))$ a.s., where
$y:=\mathbb{E}[Z(T)g(X(T))]$; see, for example,~\cite{FK08a},
Section 10. It follows that for any $p\in[0,1]$,
\[
\mathbb{P}\{Y^{y,\pi}(T) \ge g(X(T))\} = 1 \ge p.\vadjust{\goodbreak}
\]
Also observe that
\begin{eqnarray*}
\widetilde{V}(T,x,p):\!&=&\frac{V(T,x,p)}{g(x)}\\
&=&\inf\bigl\{r>0 | \exists
\pi
\in\mathcal{H} \mbox{ s.t. }
\mathbb{P} \bigl\{Y^{r g(x), \pi}(T) \ge g(X(T))\bigr\} \ge p\bigr\}.
\end{eqnarray*}
When $g(x)=\sum_{i=1}^{d}x_i$, observe that $\widetilde{V}(T,x,1)$ is
equal to equation (6.1) of~\cite{FK10OA}, the smallest relative amount
to beat the market capitalization\break $\sum_{i=1}^d X_i(T)$.

\begin{rem}
Clearly,
%
\begin{equation} \label{eqest1}\quad
0 = V(T,x,0) \le V(T,x,p) \nearrow V(T,x,1) \leq g(x)\qquad
\mbox{as } p\to1.
\end{equation}
\end{rem}

By analogy with~\cite{FL99}, we shall present a probabilistic
characterization of
$V(T,x,p)$. First, we will generalize the Neyman--Pearson lemma (see,
e.g.,~\cite{MR2169807}, Theorem A.28) in the next result.
%
\begin{lem}
\label{t-pchar}
Suppose that Assumption \ref{asfassp} holds, and $g$ satisfies (\ref
{eqasmong}). Let
$A \in\mathcal{F}_T$ satisfy
%
\begin{equation}
\label{eqt-pchar2}
\mathbb{P}(A) \geq p.
\end{equation}
Then
%
\begin{equation}
\label{eqt-pchar3}
V(T,x,p) \le\mathbb{E} [Z(T)
g(X(T)) 1_{A}].
\end{equation}
Furthermore, if $A \in\mathcal{F}_T$ satisfies (\ref{eqt-pchar2})
with equality and
%
\begin{equation}
\label{eqt-pchar1}
\mathop{\operatorname{ess}\operatorname{sup}}_{A} \{Z (T) g(X(T))\}
\le\mathop{\operatorname{ess}\operatorname{inf}}_{A^c}
\{Z (T) g(X(T))\},
\end{equation}
then $A$ satisfies (\ref{eqt-pchar3}) with equality.
\end{lem}
\begin{pf}
Under Assumption \ref{asfassp}, since $g(X(T)) 1_A \in\mathcal{F}_T$
satisfies condition~(\ref{eqasmong}), it is replicable with initial
capital $y:=\mathbb{E}[Z(T) g(X(T)) 1_A ]$; see, for example, Section
10.1 of~\cite{FK08a}. That is, there exists $\pi\in\mathcal{H}$ such
that $Y^{y,\pi}(T)=g(X(T)) 1_A$ a.s. Now if $\mathbb{P}(A) \ge p$, we
have $\mathbb{P}\{Y^{y,\pi}(T)\ge g(X(T))\} = \mathbb{P}\{1_A\ge1\}
\ge
p$. Then it follows from (\ref{eqprb}) that
$V(T,x,p) \le y = \mathbb{E}[Z(T) g(X(T)) 1_A]$.

Now, take an arbitrary pair $(y_0, \pi_0)$ of initial capital and
admissible portfolio that replicates
$g(X(T))$ with probability greater than or equal to $p$; that is,
\[
\mathbb{P}\{B\} \ge p\qquad \mbox{where } B \triangleq
\{Y^{y_0,\pi_0}(T) \ge g(X(T))\}.
\]
Let $A\in\mathcal{F}_T$ satisfy $p=\P(A)\le\P(B)$ and
(\ref{eqt-pchar1}).
To prove equality in (\ref{eqt-pchar3}), it is enough to show that
\[
y_0
\ge\mathbb{E} [Z(T) g(X(T)) 1_{A}],
\]
which can be shown as follows:
\begin{eqnarray*}
y_0 & \geq & \mathbb{E}[Z(T)Y^{y_0, \pi_0}(T)]
= \mathbb{E}[Z(T) Y^{y_0, \pi_0}(T)1_B] +
\mathbb{E}[Z(T)Y^{y_0, \pi_0}(T)1_{B^c}] \\
& \ge & \mathbb{E} [Z(T) g(X(T)) 1_B] = \mathbb{E} [Z(T) g(X(T)) 1_{A
\cap B}] + \mathbb{E}
[Z(T) g(X(T)) 1_{A^c \cap B}] \\
& \ge & \mathbb{E} [Z(T) g(X(T)) 1_{A\cap B}] +
\mathbb{P}(A^c\cap B)
\mathop{\operatorname{ess}\operatorname{inf}}_{A^c \cap B} \{Z(T) g(X(T))\}
\\
& \ge & \mathbb{E} [Z(T) g(X(T)) 1_{A\cap B}] +
\mathbb{P}(A\cap B^c) \mathop{\operatorname{ess}\operatorname{sup}}_{A\cap B^c} \{Z(T) g(X(T))\} \\
& \ge & \mathbb{E} [Z(T) g(X(T)) 1_{A \cap B}] + \mathbb{E}
[Z(T) g(X(T)) 1_{A \cap B^c}] \\
& = & \mathbb{E} [Z(T) g(X(T)) 1_{A}],
\end{eqnarray*}
where in the fourth inequality we use the following two observations: First,
$\mathbb{P}(A^c \cap B) = \mathbb{P}(A\cup B) - \mathbb{P}(A)
\ge\mathbb{P}(A\cup B) - \mathbb{P}(B) = \mathbb{P}(B^c \cap A)$.
Second,
\begin{eqnarray*}
\mathop{\operatorname{ess}\operatorname{inf}}_{A^c \cap B} \{Z(T) g(X(T))\}
&\ge&
\mathop{\operatorname{ess}\operatorname{inf}}_{A^c}
\{Z(T) g(X(T))\} \\
&\ge& \mathop{\operatorname{ess}\operatorname{sup}}_{A} \{Z(T) g(X(T))\} \\
&\ge&
\mathop{\operatorname{ess}\operatorname{sup}}_{A \cap B^c} \{Z(T) g(X(T))\},
\end{eqnarray*}
in which the second inequality follows from (\ref{eqt-pchar1}).
\end{pf}

Let $F(\cdot)$ be the cumulative distribution function of
$Z(T)g(X(T))$, and, for any $a \in\R_+$, define
\[
A_a := \{\omega\dvtx Z(T) g(X(T)) <a\},\qquad \partial A_a :=
\{\omega\dvtx Z(T) g(X(T)) =a\},
\]
and let $\bar A_a$ denote $A_a \cup\partial A_a$; that is,
%
\begin{equation}\label{barAa}
\bar A_a = \{\omega\dvtx Z(T) g(X(T)) \le a\}.
\end{equation}
Taking $A =
\bar A_a$ in Lemma \ref{t-pchar}, we see that (\ref{eqt-pchar1}) is
satisfied. It follows that
%
\begin{equation}
\label{equrep1}
V(T,x,F(a)) =\mathbb{E}[Z(T) g(X(T)) 1_{\bar A_a}].
\end{equation}
On the other hand, taking $A = A_a$, we see that (\ref{eqt-pchar1}) is
again satisfied. We therefore obtain
%
\begin{equation}
\label{equrep1-1}
V(T,x,F(a-)) = \mathbb{E}[Z(T) g(X(T)) 1_{A_a}].
\end{equation}
The last two equalities imply the following relationship:
%
\begin{eqnarray}
\label{equrep2}
V(T,x,F(a)) &=& V(T,x,F(a-)) + a \mathbb{P}\{\partial
A_a\} \nonumber\\[-8pt]\\[-8pt]
&=& V(T,x,F(a-)) + a \bigl(F(a) - F(a-)\bigr).
\nonumber
\end{eqnarray}
Next, we will determine $V(T,x,p)$ for $p \in(F(a-), F(a))$ when
\mbox{$F(a-) < F(a)$}.

\begin{prop}
\label{c-pchar} Suppose Assumption \ref{asfassp} holds. Fix an
$(x,p)\in\break
(0,\infty)^d \times[0,1]$.\vadjust{\goodbreak}
\begin{longlist}
\item There exists $A\in\mathcal{F}_T$ satisfying
(\ref{eqt-pchar2}) with equality and (\ref{eqt-pchar1}). As
a~result, (\ref{eqt-pchar3}) holds with equality.
\item If $F^{-1}(p) := \{s \in\R_+\dvtx F(s) = p\} =\varnothing$,
then letting $a: = \inf\{s \in\R_+\dvtx\break F(s) >p\}$
we have
%
\begin{eqnarray}
\label{eqc-pchar1}
V(T,x,p) & = & V(T,x,F(a-)) + a \bigl(p - F(a-)\bigr) \nonumber\\[-8pt]\\[-8pt]
& = & V(T,x,F(a)) - a\bigl(F(a) - p\bigr).
\nonumber
\end{eqnarray}
\end{longlist}
\end{prop}
\begin{pf}
(i) If there exists $a\in\mathbb{R}$ such that either $F(a) =
p$ or $F(a-) =
p$, then we can take $A=A_a$ or $A=\bar{A}_a$, thanks to (\ref
{equrep1}) and
(\ref{equrep1-1}). In the rest of the proof we will assume that
$F^{-1}(p)=\varnothing$.

Let $\widetilde{W}$ be a Brownian motion with respect to $\mathbb{F}$,
and define $B_b =
\{\omega\dvtx \frac{\wW(T)}{\sqrt{T}} <b\}$. Let us define $f(\cdot)$ by
$f(b) =
\mathbb{P}\{\partial A_a \cap B_b\}$. The function $f$ satisfies $\lim
_{b\to-\infty} f(b)
= 0$ and $\lim_{b\to\infty} f(b) = \mathbb{P}(\partial A_a)$.
Moreover, the function $f(\cdot)$ is continuous and nondecreasing. Right
continuity can be shown as follows: for $\varepsilon>0$,
\[
0\le f(b+\varepsilon) - f(b) = \mathbb{P}(\partial A_a\cap
B_{b+\varepsilon} ) - \mathbb{P}(\partial A_a \cap B_b) \le
\mathbb{P}(B_{b+\varepsilon} \cap B_b^c).
\]
The right continuity follows from observing that the last expression
goes to zero as $\varepsilon\to0$.
One can show left continuity of $f(\cdot)$ in a similar fashion.

Since $0<p - \mathbb{P}(A_a) < \mathbb{P}(\partial A_a)$, thanks to
the above properties of $f$, there
exists $b^* \in\R$ satisfying $f(b^*) = p - \mathbb{P}(A_a)$.

Define $A := A_a \cup(\partial A_a \cap B_{b^*})$. Observe that
$\mathbb{P}(A) = \mathbb{P}(A_a) + \mathbb{P}(\partial A_a \cap
B_{b^*}) = p$ and that $A$ satisfies (\ref{eqt-pchar1}).

(ii) This follows immediately from (1):
\begin{eqnarray*}
V(T,x,p) & = & \mathbb{E}[ Z(T) g(X(T))
1_A] \\
& = & \mathbb{E}[Z(T)g(X(T)) 1_{A_a}] +
\mathbb{E}[Z(T)g(X(T)) 1_{\partial A_a
\cap B_{b^*}}] \\
& = & V(T,x,F(a-)) + a \mathbb{P}(\partial A_a \cap
B_{b^*}) \\
& = & V(t,x,F(a-)) + a\bigl(p-F(a-)\bigr).
\end{eqnarray*}
\upqed\end{pf}
%
\begin{rem}\label{remrafl}
Note that when $Z$ is a martingale, using the Neyman--Pearson lemma, it
was shown in~\cite{FL99} that
%
\begin{equation}\label{eqamp}
V(T,x,p)=\inf_{\varphi\in\M}\E[Z(T) g(X(T)) \varphi] =\E[Z(T) g(X(T))
\varphi^*],
\end{equation}
where
%
\begin{equation}
\label{eqdefm}
\mathcal{M} = \{\varphi\dvtx \Omega\to[0,1] | \mathcal{F}_T
\mbox{ measurable}, \mathbb{E} [\varphi] \geq p\}.
\end{equation}
The randomized test function $\varphi^*$ is not necessarily an
indicator function. Using Lemma \ref{t-pchar} and the fine structure of
the filtration $\F_T$, we provide in Proposition~\ref{c-pchar} another
optimizer of (\ref{eqamp}) which is an indicator function.
\end{rem}

\begin{prop}
\label{c-convex}
Suppose Assumption \ref{asfassp} holds. Then, the map $p\mapsto
V(T,x,p)$ is
convex and continuous on the closed interval $[0,1]$. Hence, $V(T,x,p)
\le p V(T,x,1) \leq p g(x)$ for all $p\in[0,1]$.
\end{prop}
\begin{pf}
By Proposition \ref{c-pchar}, for any $p\in[0,1]$ there exists $A\in
\mathcal{F}_T$ such that
\[
V(T,x,p) = \mathbb{E}[Z(T)g(X(T))1_A] \le\mathbb{E}[Z(T)g(X(T))] <
\infty.
\]
Then thanks to a theorem by Ostroski (see~\cite{Donoghue}, page 12), to
show the convexity it suffices to demonstrate the midpoint convexity
%
\begin{eqnarray}
\label{eqc-convex1}
&&
\frac{V(T,x,p_1) + V(T,x,p_2)}{2} \nonumber\\[-8pt]\\[-8pt]
&&\qquad\ge V\biggl(T,x, \frac
{p_1+p_2}{2}\biggr)\qquad
\mbox{for all } 0\le p_1< p_2 \le1.\nonumber
\end{eqnarray}
Denote $\tilde p \triangleq\frac{p_1 + p_2}{2}$. It follows from
Proposition \ref{c-pchar} that there exist \mbox{$A_1 \subset\tilde A
\subset A_2$} with $\P(A_1)=p_1<\P(\tilde{A})=\tilde{p}<\P(A_2)=p_2$ satisfying
(\ref{eqt-pchar1}),
\[
V(T,x,p_i) = \mathbb{E}[Z(T)g(X(T))
1_{A_i}],\qquad i = 1,2,
\]
and
\[
V(T,x,\tilde p) =
\mathbb{E}[Z(T)g(X(T)) 1_{\tilde A}].
\]
By (\ref{eqt-pchar1}),
\begin{eqnarray*}
\mathop{\operatorname{ess}\operatorname{inf}} \{Z(T)g(X(T)) 1_{A_2 \cap\tilde A^c} \}
& \ge & \mathop{\operatorname{ess}\operatorname{inf}} \{Z(T)g(X(T)) 1_{\tilde A^c} \}\\
&\ge&
\mathop{\operatorname{ess}\operatorname{sup}} \{Z(T)g(X(T)) 1_{\tilde A}\} \\
&\ge&\mathop{\operatorname{ess}\operatorname{sup}} \{ Z(T)g(X(T)) 1_{\tilde A \cap A_1^c} \},
\end{eqnarray*}
which implies that
\[
\mathbb{E}[Z(T)g(X(T)) 1_{A_2 \cap\tilde A^c}] \ge
\mathbb{E} [Z(T)g(X (T)) 1_{\tilde A \cap A_1^c} ].
\]
As a result,
\begin{eqnarray*}
&&
\mathbb{E}[Z(T)g(X(T)) 1_{A_2}] -
\mathbb{E}[Z(T)g(X(T)) 1_{\tilde A}] \\
&&\qquad\ge\mathbb{E} [Z(T)g(X(T))
1_{\tilde A} ] -
\mathbb{E} [Z(T)g(X(T)) 1_{A_1} ],
\end{eqnarray*}
which is equivalent to (\ref{eqc-convex1}).

Now thanks to convexity, we immediately have that $p\mapsto V(T,x,p)$
is continuous on $[0,1)$. It remains to show that it is continuous
from the left at $p=1$; but this is indeed true because
\begin{eqnarray*}
\lim_{a\to\infty}V(T,x,F(a)) &=& \lim_{a\to\infty}\mathbb
{E}\bigl[Z(T)g(X(T))1_{\{Z(T)g(X(T))\le a\}}\bigr]\\
&=&\mathbb{E}[Z(T)g(X(T))] = V(T,x,1),
\end{eqnarray*}
where the second equality is due to the dominated convergence theorem.
\end{pf}

\begin{exm}
\label{e-bessel}
Consider a market with a single stock, whose dynamics follow a
three-dimensional Bessel process,
that is,
\[
d X(t) = \frac{1}{X(t)} \,dt + dW(t),\qquad X_0=x>0,
\]
and let $g(x)=x$.
In this case $Z(t) =x/X(t)$, which is the classical example for a
strict local martingale; see~\cite{Johnson-Helms}.
On the other hand, $Z(t) X(t)=x$ is a martingale. Thanks to
Proposition \ref{c-pchar} there exists a set $A \in\F_T$ with $\P
(A)=p$ such that
\[
V(T,x,p) = \mathbb{E}[Z(T) X(T) 1_{A}] = p x.
\]
\end{exm}

In~\cite{FL99}, the following result was proved when $Z$ is a
martingale. Here, we generalize this result to the case where $Z$ is
only a local martingale.

\begin{prop}\label{p-urep}
Under Assumption \ref{asfassp}
%
\begin{equation}
\label{eqp-urep1}
V(T,x,p) = \inf_{\varphi\in\mathcal{M}}
\mathbb{E}[Z(T) g(X(T)) \varphi],
\end{equation}
where $\mathcal{M}$ is defined in (\ref{eqdefm}).
\end{prop}
\begin{pf}
Thanks to Proposition \ref{c-pchar} there exists a set
$A \in\F_T$ satisfying $\mathbb{P}(A)=p$ and (\ref{eqt-pchar1}) such
that $V(T,x,p) = \mathbb{E}[Z(T) g(X(T)) 1_A] $.
Since $1_A \in\mathcal{M}$, clearly
\[
V(T,x,p) \geq
\inf_{\varphi\in\mathcal{M}} \mathbb{E}[Z(T) g(X(T))
\varphi].
\]

For the other direction, it is enough to show that for any $\varphi\in
\mathcal{M}$, we have
\[
\mathbb{E}[Z(T) g(X(T)) 1_A] \le\mathbb{E}[Z(T)
g(X(T)) \varphi].
\]
Indeed, since the left-hand side is actually $V(T,x,p)$, we can get the
desired result by taking infimum on both sides over $\varphi\in
\mathcal{M}$.

Letting $M = \mathop{\operatorname{ess}\operatorname{sup}}_A \{Z(T)
g(X(T))\}$, we observe that
\begin{eqnarray*}
&&\mathbb{E}[Z(T) g(X(T)) \varphi] - \mathbb{E}[Z(T)
g(X(T)) 1_A] \\
&&\qquad= \mathbb{E}[Z(T) g(X(T)) \varphi1_A] +
\mathbb{E}[Z(T) g(X(T)) \varphi1_{A^c}] -
\mathbb{E}[Z(T) g(X(T))1_A] \\
&&\qquad= \mathbb{E}[Z(T) g(X(T)) \varphi1_{A^c}] -
\mathbb{E}[Z(T) g(X(T)) 1_A (1 - \varphi)] \\
&&\qquad
\ge\mathop{\operatorname{ess}\operatorname{inf}}_{A^c}
\{Z(T)g(X(T))\}\mathbb{E}[\varphi
1_{A^c}]-M\mathbb{E}[1_A (1-\varphi)]\\
&&\qquad
\ge M \mathbb{E}[\varphi1_{A^c}] - M \mathbb{E}[1_A
(1-\varphi)] \qquad\mbox{[by
(\ref{eqt-pchar1})]} \\
&&\qquad= M\mathbb{E}[\varphi]-M\mathbb{E}[1_A]\geq0.
\end{eqnarray*}
\upqed\end{pf}

\subsection{\texorpdfstring{A digression: Representation of $V$ as a stochastic control
problem.}{A digression: Representation of $V$ as a stochastic control
problem}}\label{secstoc-cont}
For $p\in[0,1]$, we introduce an additional controlled state variable
%
\begin{equation}
\label{eqspp}
P^{p}_{\alpha}(s) = p + \int_{0}^{s}\alpha(r)' \,dW(r),\qquad s\in[0,T],
\end{equation}
where\vspace*{1pt} $\alpha(\cdot)$ is a $\mathbb{G}$-progressively measurable $\R
^d$-valued process satisfying the integrability condition $\int_0^{T}
|\alpha(s)|^2\,ds<\infty$ a.s. such that $P^{p}_{\alpha}$ takes values in
$[0,1]$. We will denote the class of such processes by~$\mathcal{A}$.
Note that~$\mathcal{A}$ is nonempty, as the constant control $\alpha
(\cdot)\equiv(0,\ldots,0)\in\mathbb{R}^d$ obviously lies in~$\mathcal{A}$.
The next result obtains an alternative representation for $V$ in terms
of $P^p_\alpha$.

\begin{prop}\label{t-control}
Under Assumption \ref{asfassp},
%
\begin{equation}
\label{eqcontrol}
V(T,x,p) = \inf_{\alpha\in\mathcal{A}} \mathbb{E}[Z(T) g(X(T))
P^{p}_{\alpha}(T)]<\infty.
\end{equation}
\end{prop}
\begin{pf}
The finiteness follows from (\ref{eqasmong}).
Define
\[
\widetilde{\M}:= \{\varphi\dvtx \Omega\to[0,1] | \mathcal{F}_T
\mbox{ measurable}, \mathbb{E} [\varphi] = p\}.
\]
Thanks to Proposition \ref{c-pchar}, there exists a set
$A \in\F_T$ satisfying $\mathbb{P}(A)=p$ and~(\ref{eqt-pchar1})
such that
\[
V(T,x,p) = \mathbb{E}[Z(T) g(X(T)) 1_A] \ge\inf_{\varphi\in
\widetilde
{\M}}
\mathbb{E}[Z(T) g(X(T)) \varphi].
\]
Since the opposite inequality follows immediately from
Proposition \ref{p-urep}, we conclude that
\[
V(T,x,p) = \inf_{\varphi\in\widetilde{\M}}
\mathbb{E}[Z(T) g(X(T)) \varphi].
\]

Therefore, it is enough to show that $\widetilde{\M}$ satisfies
$\widetilde{\M} = \{ P^{p}_{\alpha}(T) | \alpha\in
\mathcal
{A}\}$. The inclusion $\widetilde{\M}\supset\{
P^{p}_{\alpha
}(T) | \alpha\in
\mathcal{A}\}$ is clear.
To show the other inclusion we will use the Martingale representation
theorem: for any $\varphi\in\widetilde{\M}$ there exists a $\mathbb
{G}$-progressively measurable $\R^d$-valued process $\psi(\cdot)$
satisfying $\int_0^{T} |\psi(s)|^2\,ds<\infty$ a.s. such that
\[
\mathbb{E}[\varphi|\mathcal{F}_t] = p + \int_0^t \psi(s)' \,d
W(s),\qquad
t\in[0,T].
\]
Note that since $\varphi$ takes values in $[0,1]$, so does $\mathbb
{E}[\varphi|\mathcal{F}_t]$ for all $t\in[0,T]$. Then we see that
$\mathbb{E}[\varphi|\mathcal{F}_t]$ satisfies (\ref{eqspp}) with
$\alpha(\cdot)=\psi(\cdot)\in\mathcal{A}$.
\end{pf}

\section{\texorpdfstring{The PDE characterization.}{The PDE characterization}}\label{secPDE-characterization}
\subsection{\texorpdfstring{Notation.}{Notation}}
We denote by $X^{t,x}(\cdot)$ the solution of (\ref{eqstk}) starting
from~$x$ at time $t$ and by $Z^{t,x,z}(\cdot)$ the solution of
%
\begin{equation}
\label{eqdeflate}
d Z(s) = - Z(s) \theta(X^{t,x}(s))' \,dW(s),\qquad Z(t) = z.\vadjust{\goodbreak}
\end{equation}
Define the process $Q^{t,x,q}(\cdot)$ by
%
\begin{equation}
Q^{t,x,q}(\cdot) := \frac{1}{Z^{t,x,(1/q)}(\cdot)},\qquad q\in(0,\infty).
\end{equation}
Then we see from (\ref{eqdeflate}) that $Q(\cdot)$ satisfies
%
\begin{equation}\label{Qdynamics}
\quad\frac{dQ(s)}{Q(s)}=|\theta(X^{t,x}(s))|^{2}\,ds+\theta
(X^{t,x}(s))'\,dW(s),\qquad Q^{t,x,q}(t)=q.
\end{equation}
We then introduce the value function
\[
U(t,x,p) := \inf_{\varphi\in\mathcal{M}} \mathbb{E}[Z^{t,x,1}(T)
g(X^{t,x}(T))\varphi],
\]
where $\mathcal{M}$ is defined in (\ref{eqdefm}). Note that the
original value function $V$ can be written in terms of $U$ as
$V(T,x,p)=U(0,x,p)$.

We also consider the Legendre transform of $U$ with respect to the $p$
variable. To make the discussion clear, however, let us first extend
the domain of the map $p\mapsto U(t,x,p)$ from $[0,1]$ to the entire
real line $\mathbb{R}$ by setting
%
\begin{eqnarray}
\label{eqextU1}
U(t,x,p) &=& 0 \qquad\mbox{for } p<0,\\
\label{eqextU2}
U(t,x,p) &=& \infty\qquad\mbox{for } p>1.
\end{eqnarray}
Then the Legendre transform of $U$ with respect to $p$ is well defined:
%
\begin{eqnarray}\label{eqw=sup}
w(t,x,q):\!&=&\sup_{p\in\mathbb{R}}\{pq-U(t,x,p)\}\nonumber\\[-8pt]\\[-8pt]
&=&\cases{
\infty, &\quad if $q<0$;\cr
\displaystyle \sup_{p\in[0,1]}\{pq-U(t,x,p)\}, &\quad if $q\ge0$.}\nonumber
\end{eqnarray}
From Proposition \ref{c-convex}, we already know that $p\mapsto
U(t,x,p)$ is convex and continuous on $[0,1]$. Since $U(t,x,0)=0$, we
see from (\ref{eqextU1}) and (\ref{eqextU2}) that $p\mapsto
U(t,x,p)$ is continuous on $(-\infty,1]$ and lower semicontinuous on
$\mathbb{R}$. Moreover, considering that $p\mapsto U(t,x,p)$ is
increasing on $[0,1]$, we conclude that $p\mapsto U(t,x,p)$ is also
convex on $\mathbb{R}$. Now thanks to~\cite{VT}, Section 6.18, the convexity
and the lower semicontinuity of $p\mapsto U(t,x,p)$ on $\mathbb{R}$
imply that the double transform of $U$ is indeed equal to $U$ itself.
That is, for any $(t,x,p)\in[0,T]\times(0,\infty)^d\times\mathbb{R}$,
\[
U(t,x,p) = \sup_{q\in\mathbb{R}}\{pq-w(t,x,q)\} = \sup_{q\ge0}\{
pq-w(t,x,q)\},
\]
where the second equality is a consequence of (\ref{eqw=sup}).

In this section, we also consider the function
%
\begin{eqnarray}
\label{eqdeftildew}
\widetilde{w}(t,x,q):\!&=& \mathbb
{E}\bigl[Z^{t,x,1}(T)\bigl(Q^{t,x,q}(T)-g(X^{t,x}(T))\bigr)^+\bigr]\nonumber\\[-8pt]\\[-8pt]
&=&\mathbb
{E}\bigl[\bigl(q-Z^{t,x,1}(T)g(X^{t,x}(T))\bigr)^+\bigr]\nonumber
\end{eqnarray}
for any $(t,x,q)\in[0,T]\times(0,\infty)^d\times(0,\infty)$. We will
show that $w=\widetilde{w}$ and derive various properties of
$\widetilde{w}$.
%
\begin{rem}
From the definition of $\widetilde{w}$ in (\ref{eqdeftildew}),
$\widetilde{w}$ is the upper hedging price for the contingent claim
$(Q^{t,x,q}(T)-g(X^{t,x}(T)))^+$, and potentially solves the linear PDE
%
\begin{equation}\label{PDEdegenerate}
\partial_t \widetilde{w} + \tfrac{1}{2}\operatorname{Tr}(\sigma
\sigma'D^2_x
\widetilde{w}) + \tfrac{1}{2}|\theta|^2 q^2 D^2_q \widetilde{w} + q
\operatorname{Tr}(\sigma\theta D_{xq}\widetilde{w}) = 0.
\end{equation}
This is not, however, a traditional Black--Scholes-type equation
because it is degenerate on the entire space $(x,q)\in(0,\infty
)^d\times
(0,\infty)$. Consider the following function $v$ which takes values in
the space of $(d+1)\times d$ matrices:
\[
v(\cdot):=\left[
\matrix{
s(\cdot)_{d\times d} \cr\hline
\theta(\cdot)'_{1\times d}}\right].
\]
Degeneracy can be seen by observing that $v(x)v(x)'$ is only positive
semi-definite for all $x\in(0,\infty)^d$. Or, one may observe
degeneracy by noting that there are $d+1$ risky assets, $X_1,\ldots
,X_d$ and $Q$, with only $d$ independent sources of uncertainty,
$W_1,\ldots,W_d$. As a result, the existence of classical solutions to
(\ref{PDEdegenerate}) cannot be guaranteed by standard results for
parabolic equations. Indeed, under the setting of Example \ref
{e-bessel}, we have
\[
\widetilde{w}(t,x,q)= \mathbb{E}\bigl[\bigl(q-Z^{t,x,1}(T)X^{t,x}(T)\bigr)^+\bigr]=(q-x)^+,
\]
which is not smooth.
\end{rem}

\subsection{\texorpdfstring{Elliptic regularization.}{Elliptic
regularization}}
In this subsection, we will approximate $\widetilde{w}$ by a sequence
of smooth functions $\widetilde{w}_\varepsilon$, constructed by
elliptic regularization. We will then derive some properties of
$\widetilde{w}_\varepsilon$ and investigate the relation between~$\widetilde{w}$
and $\widetilde{w}_\varepsilon$. Finally, we will show
that $\widetilde{w}=w$, which validates the construction of
$\widetilde
{w}_\varepsilon$.

To perform elliptic regularization under our setting, we need to first
introduce a~product probability space. Recall that we have been working
on a~probability space $(\Omega,\mathbb{F},\mathbb{P})$, given by a
weak solution to the SDE (\ref{eqstk}). Now consider the sample space
$\Omega^B:=C([0,T];\mathbb{R})$ and the canonical process~$B(\cdot)$.
Let $\mathbb{F}^B$ be the filtration generated by $B$ and $\mathbb
{P}^B$ be the Wiener measure on $(\Omega^B,\mathbb{F}^B)$. We then
introduce the product probability space $(\bar{\Omega},\bar{\mathbb
{F}},\bar{\mathbb{P}})$, with $\bar{\Omega}:=\Omega\times\Omega^B$,
$\bar{\mathbb{F}}:=\mathbb{F}\times\mathbb{F}^B$ and $\bar
{\mathbb
{P}}:=\mathbb{P}\times\mathbb{P}^B$. For any\vspace*{1pt} $\bar{\omega}\in\bar
{\Omega
}$, we write $\bar{\omega}=(\omega,\omega^B)$, where $\omega\in
\Omega$
and $\omega^B\in\Omega^B$. Also, we denote by $\bar{\mathbb{E}}$ the
expectation taken under $(\bar{\Omega},\bar{\mathbb{F}},\bar
{\mathbb{P}})$.

For any $\varepsilon>0$, introduce the process $Q^{t,x,q}_\varepsilon
(\cdot)$ which satisfies the following dynamics:
%
\begin{eqnarray}\label{eqdQepsilon}
\frac{dQ_\varepsilon(s)}{Q_\varepsilon(s)} &=& |\theta
(X^{t,x}(s))|^2\,ds +
\theta(X^{t,x}(s))'\,dW(s) + \varepsilon \,dB(s),\nonumber\\[-8pt]\\[-8pt]
Q^{t,x,q}_\varepsilon
&=&q\in(0,\infty).\nonumber
\end{eqnarray}
Then under the probability space $(\bar{\Omega},\bar{\mathbb
{F}},\bar
{\mathbb{P}})$, we have $d+1$ risky assets, the~$d$ stocks $X_1,\ldots
,X_d$ and $Q_\varepsilon$. Define
\[
\bar{s}:=\left[\begin{array}{c@{\quad}c@{\quad}c@{\hspace*{4pt}}|@{\hspace*{4pt}}c}
s_{11}&\cdots &s_{1d}&0\\
\vdots  &\ddots &\vdots &\vdots\\
s_{d1}&\cdots &s_{dd}&0\\ \hline
\theta_{1}&\cdots &\theta_{d}&\varepsilon
\end{array}\right],\qquad
\bar{b}:= \left[\begin{array}{c}
b_{1}\\
\vdots\\
b_{d}\\ \hline
|\theta|^2
\end{array}\right]
\]
and
\[
\bar{a}:=\bar{s}\bar{s}'=
\left[\begin{array}{c@{\quad}c@{\quad}c@{\hspace*{4pt}}|@{\hspace*{4pt}}c}
a_{11} &\cdots & a_{1d} &|\\
\vdots &\ddots &\vdots & s\theta\\[2pt]
a_{d1} &\cdots &a_{dd} &|\\ \hline
- &\theta' s' &- & |\theta|^2+\varepsilon^2\\
\end{array}\right].
\]
Since we assume that the matrix $s$ has full rank (Assumption \ref
{asfassp}), $\bar{s}$ has full rank by definition. It follows that
$\bar{a}$ is positive definite. Now we can define the corresponding
market\vspace*{1pt} price of risk under $(\bar{\Omega},\bar{\mathbb{F}},\bar
{\mathbb
{P}})$ as $\bar{\theta}:=\bar{s}^{-1}\bar{b}$, and the corresponding
deflator $\bar{Z}(\cdot)$ under $(\bar{\Omega},\bar{\mathbb
{F}},\bar
{\mathbb{P}})$ as the solution of
%
\begin{equation}\label{eqbardeflate}
d\bar{Z}(s) = -\bar{Z}(s)\bar{\theta}(X^{t,x}(s))'\,d\bar{W}(s),\qquad
\bar{Z}^{t,x,z}(t)=z,
\end{equation}
where $\bar{W}:=(W_1,\ldots,W_d,B)$ is a $(d+1)$-dimensional Brownian
motion. Observe that
\[
\bar{\theta}
= \left[\begin{array}{c@{\hspace*{4pt}}|@{\hspace*{4pt}}c}
s^{-1}&O_{d\times 1}\\ \hline
-\dfrac{1}{\varepsilon}\theta's^{-1}&\dfrac{1}{\varepsilon}
\end{array}\right]
\left[\begin{array}{c}
b\\
|\theta|^2
\end{array}\right]
=\left[\begin{array}{c}
\theta\\
0
\end{array}\right].
\]
This implies that (\ref{eqbardeflate}) coincides with (\ref
{eqdeflate}). Thus, we conclude that $\bar{Z}(\cdot)=Z(\cdot)$.
Finally, let us introduce the function
\[
\widetilde{w}_\varepsilon(t,x,q):=\bar{\mathbb{E}}\bigl[\bar
{Z}^{t,x,1}(T)\bigl(Q^{t,x,q}_\varepsilon(T)-g(X^{t,x}(T))\bigr)^+\bigr]
\]
for any $(t,x,q)\in[0,T]\times(0,\infty)^d\times(0,\infty)$. By
(\ref{eqdQepsilon}) and (\ref{Qdynamics}), we see that the processes
$Q_\varepsilon(\cdot)$ and $Q(\cdot)$ have the following relation:
%
\begin{eqnarray}\label{QQe}
Q^{t,x,q}_\varepsilon(s) = Q^{t,x,q}(s)\exp\bigl\{-\tfrac
{1}{2}\varepsilon^2(s-t)+\varepsilon\bigl(B(s)-B(t)\bigr)\bigr\},\nonumber\\[-8pt]\\[-8pt]
&&\eqntext{s\in[t,T].}
\end{eqnarray}
It then follows from (\ref{QQe}), the fact that $\bar{Z}(\cdot
)=Z(\cdot
)$ and the definition of~$\widetilde{w}_\varepsilon$ that
%
\begin{eqnarray}\label{twewithnoZ}
\widetilde{w}_\varepsilon(t,x,q)&=&\bar{\mathbb{E}}\bigl[
\bigl(q\exp
\bigl\{-\tfrac{1}{2}\varepsilon^2(T-t)+\varepsilon\bigl(B(T)-B(t)\bigr)\bigr\}\nonumber\\[-8pt]\\[-8pt]
&&\hspace*{84.5pt}{}-Z^{t,x,1}(T)g(X^{t,x}(T))\bigr)^+\bigr].\nonumber
\end{eqnarray}

%
\begin{asm}\label{aslocLips}
The functions $\theta_i$ and $s_{ij}$ are locally Lipschitz, for all
$i,j\in\{1,\ldots,d\}$.
\end{asm}
%
\begin{lem}\label{lemsmoothtwe}
Under Assumption \ref{aslocLips}, we have that $\widetilde
{w}_\varepsilon\in\mathcal{C}^{1,2,2}((0,T)\times(0,\infty
)^{d}\times
(0,\infty))$ and satisfies the PDE
%
\begin{equation}\label{PDEtwe}\quad
\partial_t \widetilde{w}_\varepsilon+ \tfrac{1}{2}\operatorname
{Tr}(\sigma\sigma
'D^2_x \widetilde{w}_\varepsilon) + \tfrac{1}{2}(|\theta
|^2+\varepsilon
^2)q^2D^2_q \widetilde{w}_\varepsilon+ q \operatorname{Tr}(\sigma
\theta
D_{xq}\widetilde{w}_\varepsilon) = 0,\hspace*{-32pt}
\end{equation}
$(t,x,q)\in(0,T)\times(0,\infty)^d\times(0,\infty)$, with the
boundary condition
%
\begin{equation}\label{boundarytwe}
\widetilde{w}_\varepsilon(T,x,q)=\bigl(q-g(x)\bigr)^+.
\end{equation}
\end{lem}
\begin{pf}
Since $\bar{a}$ is positive definite and continuous, it must satisfy
the following ellipticity condition: for every compact set $K\subset
(0,\infty)^d$, there exists a positive constant $C_K$ such that
%
\begin{equation}\label{ellipticity}
\sum_{i=1}^{d+1}\sum_{j=1}^{d+1}\bar{a}_{ij}(x)\xi_i\xi_j \ge
C_K|\xi|^2
\end{equation}
for all $\xi\in\mathbb{R}^{d+1}$ and $x\in K$; see, for example,~\cite{HS},
Lemma 3. Under Assumption \ref{aslocLips} and (\ref{ellipticity}),
the smoothness of $\widetilde{w}_\varepsilon$ and the PDE (\ref
{PDEtwe}) follow immediately from~\cite{Ruf2010}, Theorem 4.2.
Finally, note
that $\widetilde{w}_\varepsilon$ satisfies the boundary condition by
definition.
\end{pf}
%
\begin{prop}\label{propstrictconvex}
For any $(t,x)\in[0,T]\times(0,\infty)^d$, the map $q\mapsto
\widetilde
{w}_\varepsilon(t, x,q)$ is strictly convex on $(0,\infty)$. More
precisely, the map $q\mapsto D_q\widetilde{w}_\varepsilon(t,x,q)$ is
strictly increasing on $(0,\infty)$ with
\[
\lim_{q\downarrow0}D_q\widetilde{w}_\varepsilon(t,x,q)=0
\quad\mbox
{and}\quad
\lim_{q\to\infty}D_q\widetilde{w}_\varepsilon(t,x,q)=1.
\]
\end{prop}
\begin{pf}
We will first compute $D_q\widetilde{w}_\varepsilon(t,x,q)$, and then
show that it is strictly increasing in $q$ from $0$ to $1$. Let
$L_\varepsilon(t,T) := \exp(-\frac{1}{2}\varepsilon
^2(T-t)+\varepsilon(B(T)-B(t)))$ and $\widetilde{A}_a := \{
\bar
{\omega}\dvtx Z^{t,x,1}(T)g(X^{t,x}(T))\le aL_\varepsilon(t,T)\}$ for
$a\ge
0$. Fix an arbitrary \mbox{$q>0$}. For any $\delta>0$, define
\[
E^\delta:=\{\bar{\omega}\dvtx qL_\varepsilon(t,T) <
Z^{t,x,1}(T)g(X^{t,x}(T)) \le(q+\delta)L_\varepsilon(t,T)\}.
\]
Note that by construction, $\widetilde{A}_q$ and $E^\delta$ are
disjoint, and $\widetilde{A}_{q+\delta}=\widetilde{A}_q\cup E^\delta$.
It follows that
\begin{eqnarray*}
&&\frac{1}{\delta}[\widetilde{w}_\varepsilon
(t,x,q+\delta
)-\widetilde{w}_\varepsilon(t,x,q)]\\
&&\qquad= \frac{1}{\delta}\bigl\{\bar{\mathbb{E}}\bigl[\bigl((q+\delta
)L_\varepsilon(t,T)-Z^{t,x,1}(T)g(X^{t,x}(T))\bigr)1_{\widetilde
{A}_{q+\delta}}\bigr]\\
&&\qquad\quad\hspace*{34pt}{}-\bar{\mathbb{E}}\bigl[\bigl(qL_\varepsilon
(t,T)-Z^{t,x,1}(T)g(X^{t,x}(T))\bigr)1_{\widetilde{A}_q}
\bigr]\bigr\}\\
&&\qquad=
\frac{1}{\delta}\bigl\{\bar{\mathbb{E}}\bigl[\bigl((q+\delta
)L_\varepsilon(t,T)-Z^{t,x,1}(T)g(X^{t,x}(T))\bigr)1_{\widetilde
{A}_q}\bigr]\\
&&\qquad\quad\hspace*{9.7pt}{}+\bar{\mathbb{E}}\bigl[\bigl((q+\delta)L_\varepsilon
(t,T)-Z^{t,x,1}(T)g(X^{t,x}(T))\bigr)1_{E^\delta}\bigr]\\
&&\qquad\quad\hspace*{35.3pt}{} -\bar{\mathbb{E}}\bigl[\bigl(qL_\varepsilon
(t,T)-Z^{t,x,1}(T)g(X^{t,x}(T))\bigr)1_{\widetilde{A}_q}\bigr]\bigr\}
\\
&&\qquad=
\bar{\mathbb{E}}[L_\varepsilon(t,T)1_{\widetilde{A}_q}]+\frac
{1}{\delta
}\bar{\mathbb{E}}\bigl[\bigl((q+\delta)L_\varepsilon
(t,T)-Z^{t,x,1}(T)g(X^{t,x}(T))\bigr)1_{E^\delta}\bigr].
\end{eqnarray*}
By the definition of $E^\delta$,
\begin{eqnarray*}
0 & \leq& \frac{1}{\delta}\bar{\mathbb{E}}\bigl[\bigl((q+\delta
)L_\varepsilon(t,T)-Z^{t,x,1}(T)g(X^{t,x}(T))\bigr)1_{E^\delta
}\bigr]\\
&\leq&\frac{1}{\delta}\bar{\mathbb{E}}[\delta L_\varepsilon(t,T)
1_{E^\delta}]\\
& = & \bar{\mathbb{E}}[L_\varepsilon(t,T)1_{E^\delta}]\to0\qquad \mbox
{as } \delta\downarrow0,
\end{eqnarray*}
where we use the dominated convergence theorem. We therefore conclude that
\[
D_q\widetilde{w}_\varepsilon(t,x,q) = \lim_{\delta\downarrow
0}{\frac
{1}{\delta}[\widetilde{w}_\varepsilon(t,x,q+\delta)-\widetilde
{w}_\varepsilon(t,x,q)]}=\bar{\mathbb{E}}[L_\varepsilon
(t,T)1_{\widetilde{A}_q}].
\]
Thanks to the dominated convergence theorem again, we have
\[
\lim_{q\downarrow0}\bar{\mathbb{E}}[L_\varepsilon
(t,T)1_{\widetilde
{A}_q}]=0 \quad\mbox{and}\quad \lim_{q\to\infty}\bar{\mathbb
{E}}[L_\varepsilon
(t,T)1_{\widetilde{A}_q}]=\bar{\mathbb{E}}[L_\varepsilon(t,T)]=1.
\]

It remains\vspace*{-1pt} to prove that $D_q\widetilde{w}_\varepsilon(t,x,q)=\bar
{\mathbb{E}}[L_\varepsilon(t,T)1_{\widetilde{A}_q}]$ is strictly
increasing in $q$. Note that it is enough to show that the event
$E^\delta$ has positive probability for all $\delta>0$. Under the
integrability condition (\ref{eqass-mrp}), the deflator~$Z(\cdot)$ is
strictly positive with probability 1; see, for example,~\cite{B10},
Section~6. It follows from our assumptions on $g$ [see
(\ref{eqasmong}) and the line before it] that
\[
0<Z^{t,x,1}(T)g(X^{t,x}(T))<\infty, \qquad\mathbb{P}\mbox{-a.s.}
\]
This implies that
%
\begin{equation}\label{eqrangeZgX}
-\infty< \log Z^{t,x,1}(T)g(X^{t,x}(T)) < \infty,\qquad \bar{\P}\mbox{-a.s.}
\end{equation}
Now, from (\ref{eqrangeZgX}) and the definitions $E^\delta$ and
$L_\eps$, we see that $\bar{\P}(E^\delta)$ equals to the
probability of
the event
\begin{eqnarray*}
&&\biggl\{\bar{\omega}\dvtx\frac{\eps}{2}(T-t)+\frac{1}{\eps}\log\frac
{Z^{t,x,1}(T)g(X^{t,x}(T))}{q +\delta} \\
&&\qquad\le B(T)-B(t)<\frac{\eps
}{2}(T-t)+\frac{1}{\eps}\log\frac{Z^{t,x,1}(T)g(X^{t,x}(T))}{q}
\biggr\}.
\end{eqnarray*}
%
Thanks to Fubini's theorem, this probability is strictly positive.
\end{pf}

We investigate the relation between $\widetilde{w}$ and $\widetilde
{w}_\varepsilon$ in the following result.
%
\begin{lem}\label{lemtwtwe}
The functions $\widetilde{w}$ and $\widetilde{w}_\varepsilon$ satisfy
the following relations:
\begin{longlist}
\item For any $(t,x,q)\in[0,T]\times(0,\infty)^d\times
(0,\infty)$,
\[
\widetilde{w}(t,x,q) = \lim_{\varepsilon\downarrow0} \widetilde
{w}_\varepsilon(t,x,q).
\]
\item For any compact subset $E\subset(0,\infty)$, $\widetilde
{w}_\varepsilon$ converges to $\widetilde{w}$ uniformly on
$[0,T]\times
(0,\infty)^d\times E$. Moreover, for any $(t,x,q)\in[0,T]\times
(0,\infty)^d\times(0,\infty)$,
%
\begin{equation}\label{tw=limtwe}
\widetilde{w}(t,x,q) = \lim_{(\varepsilon,t',x',q')\to(0,t,x,q)}
\widetilde{w}_\varepsilon(t',x',q').
\end{equation}
\end{longlist}
\end{lem}
\begin{pf}
(i)
By (\ref{QQe}), we observe that
%
\begin{eqnarray}\label{dominated}
&&\bar{\mathbb{E}}\Bigl[\sup_{\varepsilon\in
(0,1]}Z^{t,x,1}(T)Q^{t,x,q}_\varepsilon(T)\Bigr]\nonumber\\
&&\qquad=\bar{\mathbb
{E}}
\biggl[\sup_{\varepsilon\in(0,1]}q\exp\biggl\{-\frac{1}{2}\varepsilon
^2(T-t)+\varepsilon\bigl(B(T)-B(t)\bigr)\biggr\}\biggr]\nonumber\\
&&\qquad\le
q \bar{\mathbb{E}}\Bigl[\sup_{\varepsilon\in(0,1]}\exp\bigl\{
\varepsilon\bigl(B(T)-B(t)\bigr)\bigr\}\Bigr]\nonumber\\
&&\qquad\le
q \bar{\mathbb{E}}\Bigl[\sup_{\varepsilon\in(0,1]}\exp\bigl\{
\varepsilon\bigl(B(T)-B(t)\bigr)\bigr\}1_{\{B(T)-B(t)\ge0\}}\Bigr]\\
&&\qquad\quad{}+ q \bar
{\mathbb{E}}\Bigl[\sup_{\varepsilon\in(0,1]}\exp\bigl\{
\varepsilon
\bigl(B(T)-B(t)\bigr)\bigr\}1_{\{B(T)-B(t)< 0\}}\Bigr]\nonumber\\
&&\qquad\le q\bar{\mathbb{E}}[\exp\{B(T)-B(t)\}]+q\nonumber\\
&&\qquad= q\biggl(\exp\biggl\{\frac{1}{2}(T-t)\biggr\}+1\biggr) <
\infty.\nonumber
\end{eqnarray}
Then it follows from the dominated convergence theorem that
\begin{eqnarray*}
\lim_{\varepsilon\downarrow0}\widetilde{w}_\varepsilon
(t,x,q)&=&\lim
_{\varepsilon\downarrow0}\bar{\mathbb{E}}\biggl[\biggl(q\exp
\biggl\{-\frac
{1}{2}\varepsilon^2(T-t)+\varepsilon\bigl(B(T)-B(t)\bigr)\biggr\}\\
&&\hspace*{109.5pt}{}-Z^{t,x,1}(T)g(X^{t,x}(T))\biggr)^+\biggr]\nonumber\\
&=& \bar{\mathbb{E}}\bigl[\bigl(q-Z^{t,x,1}(T)g(X^{t,x}(T))\bigr)^+\bigr]\nonumber\\
&=& \mathbb{E}\bigl[\bigl(q-Z^{t,x,1}(T)g(X^{t,x}(T))\bigr)^+\bigr]\\
&=& \widetilde{w}(t,x,q),
\end{eqnarray*}
where the third equality is due to the fact that
$Z^{t,x,1}(T)g(X^{t,x}(T))$ depends only on $w\in\Omega$.

(ii) From (\ref{eqdeftildew}), (\ref{twewithnoZ}) and the
observation that $|(a-b)^+ -(c-b)^+|\le|a-c|$ for any $a,b,c\in
\mathbb{R}$,
%
\begin{eqnarray}\label{UC1}
&&
|\widetilde{w}_\varepsilon(t,x,q)-\widetilde{w}(t,x,q)|\nonumber\\
&&\qquad\le q\bar
{\mathbb{E}}\biggl|\exp\biggl\{-\frac{1}{2}\varepsilon
^2(T-t)+\varepsilon
\bigl(B(T)-B(t)\bigr)\biggr\}-1\biggr|\nonumber\\
&&\qquad\le q\bar{\E}\biggl[\exp\biggl\{\frac{\eps^2}{2}(T-t)+\eps
|B(T)-B(t)|\biggr\}-1\biggr]\\
&&\qquad= q\bigl[\bigl(1+\Phi\bigl(\eps\sqrt{T-t}\bigr)-\Phi\bigl(-\eps\sqrt
{T-t}\bigr)\bigr)
e^{\eps^2 (T-t)}-1\bigr]\nonumber\\
&&\qquad\le q\bigl[\bigl(1+\Phi\bigl(\eps\sqrt{T}\bigr)-\Phi\bigl(-\eps\sqrt
{T}\bigr)\bigr)
e^{\eps^2 T}-1\bigr],\nonumber
\end{eqnarray}
%
where $\Phi(\cdot)$ is the cumulative distribution function of the
standard normal
distribution. Note that the second line of (\ref{UC1}) follows from the
inequality $|e^{v}-1|\le e^{|v|}-1$ for $v\in\R$; this inequality holds
because if $v<0$, $|e^v-1|=1-e^v=(e^{-v}-1)e^v\le e^{-v}-1= e^{|v|}-1$
and if $v\ge0$, $|e^v-1|=e^v-1=e^{|v|}-1$. We can then conclude from
(\ref{UC1}) that $\widetilde{w}_\varepsilon$ converges to
$\widetilde
{w}$ uniformly on $[0,T]\times(0,\infty)^d\times E$, for any compact
subset $E$ of $(0,\infty)$. Now, by
Lemma~\ref{lemsmoothtwe}~$\widetilde{w}_\varepsilon$ is continuous on $(0,T)\times(0,\infty
)^d\times(0,\infty)$. Then as a result of uniform convergence,
$\widetilde{w}$ must be continuous on the same domain. Noting
that
\begin{eqnarray*}
|\widetilde{w}_\varepsilon(t',x',q')-\widetilde{w}(t,x,q)|&\le&
|\widetilde{w}_\varepsilon(t',x',q')-\widetilde
{w}(t',x',q')|\\
&&{}+|\widetilde{w}(t',x',q')-\widetilde{w}(t,x,q)|,
\end{eqnarray*}
we see that (\ref{tw=limtwe}) follows from the continuity of
$\widetilde{w}$ and the uniform convergence of $\widetilde
{w}_\varepsilon$ to $\widetilde{w}$ on $[0,T]\times(0,\infty
)^d\times
E$ for any compact subset $E$ of $(0,\infty)$.
\end{pf}

Thanks to the stability of viscosity solutions, we have the following
result immediately.
%
\begin{prop}\label{propviscositytw}
Under Assumption \ref{aslocLips}, we have that $\widetilde{w}$ is a
continuous viscosity solution to
%
\begin{equation}\label{PDEtw}
\partial_t \widetilde{w} + \tfrac{1}{2}\operatorname{Tr}(\sigma
\sigma'D^2_x \widetilde
{w}) + \tfrac{1}{2}|\theta|^2 q^2 D^2_q \widetilde{w} + q
\operatorname{Tr}(\sigma
\theta D_{xq}\widetilde{w}) = 0
\end{equation}
for $(t,x,q)\in(0,T)\times(0,\infty)^d\times(0,\infty)$, with the
boundary condition
%
\begin{equation}\label{boundarytw}
\widetilde{w}(T,x,q)=\bigl(q-g(x)\bigr)^+.
\end{equation}
\end{prop}
\begin{pf}
By Lemmas \ref{lemsmoothtwe} and \ref{lemtwtwe}(ii), the
viscosity solution property follows as a direct application of
\cite{Touzi-note-Pisa}, Proposition 2.3, and the boundary condition holds
trivially from the definition of $\widetilde{w}$.
\end{pf}

Now we want to relate to $\widetilde{w}$ to $w$. Given $(t,x)\in
[0,T]\times(0,\infty)^d$, recall the notation in Section\vadjust{\goodbreak} \ref{secquantile-hedging}: for any $a
\ge0$, $\bar{A}_a := \{\omega\dvtx Z^{t,x,1}(T)g(X^{t,x}(T)) \le a\}$;
also,~$F(\cdot)$ again denotes the cumulative distribution function
of\break
$Z^{t,x,1}(T)g(X^{t,x}(T))$. We first present another representation
for $\widetilde{w}$ as \mbox{follows}.

\begin{lem}\label{lemmax=tw}
For any $(t,x,q)\in[0,T]\times(0,\infty)^{d}\times(0,\infty)$, we have
\[
\max_{a\ge0}\mathbb{E}\bigl[\bigl(q-Z^{t,x,1}(T)g(X^{t,x}(T))\bigr)1_{\bar
{A}_a}\bigr]=\widetilde{w}(t,x,q).
\]
\end{lem}
\begin{pf}
Let us first take $a<q$. Since $\bar{A}_a\subset\bar{A}_q$ and
$q-Z^{t,x,1}(T)\times g(X^{t,x}(T))\ge0$ on $\bar{A}_q$,
\begin{eqnarray*}
\mathbb{E}\bigl[\bigl(q-Z^{t,x,1}(T)g(X^{t,x}(T))\bigr)1_{\bar{A}_a}\bigr]&\le&\mathbb
{E}\bigl[\bigl(q-Z^{t,x,1}(T)g(X^{t,x}(T))\bigr)1_{\bar{A}_q}\bigr]\\
&=&\widetilde{w}(t,x,q).
\end{eqnarray*}

Now consider $a>q$. Set $F:=\{\omega\dvtx q<Z^{t,x,1}(T)g(X^{t,x}(T)) \le
a\}
$. Observing that $\bar{A}_q$ and $F$ are disjoint, and $\bar
{A}_a=\bar
{A}_q\cup F$, we have
\begin{eqnarray*}
&&\mathbb{E}\bigl[\bigl(q-Z^{t,x,1}(T)g(X^{t,x}(T))\bigr)1_{\bar{A}_a}\bigr]\\
&&\qquad=\mathbb{E}\bigl[\bigl(q-Z^{t,x,1}(T)g(X^{t,x}(T))\bigr)1_{\bar{A}_q}\bigr]+\mathbb
{E}\bigl[\bigl(q-Z^{t,x,1}(T)g(X^{t,x}(T))\bigr)1_{F}\bigr]\\
&&\qquad\le \mathbb{E}\bigl[\bigl(q-Z^{t,x,1}(T)g(X^{t,x}(T))\bigr)1_{\bar{A}_q}\bigr] =
\widetilde{w}(t,x,q),
\end{eqnarray*}
where the inequality is due to the fact that
$q-Z^{t,x,1}(T)g(X^{t,x}(T))<0$ on~$F$.
\end{pf}

Next, we will argue that $w$ and $\widetilde{w}$ are equal.
%
\begin{prop}\label{propw=tw}
$w(t,x,q)=\widetilde{w}(t,x,q)$, for all $(t,x,q)\in[0,T]\times\break
(0,\infty
)^d\times(0,\infty)$.
\end{prop}
\begin{pf}
Given $p\in[0,1]$, there exists $a\ge0$ such that $F(a-)\le p\le
F(a)$. We can take two nonnegative numbers $\lambda_1$ and $\lambda_2$
with $\lambda_1+\lambda_2 = 1$ such that
%
\begin{equation}\label{plambda1,2}
p=\lambda_1F(a)+\lambda_2F(a-).
\end{equation}
Observe that $p-F(a-)=\lambda_1(F(a)-F(a-))$. Plugging this into the
first line of (\ref{eqc-pchar1}), we get
%
\begin{equation}\label{UFaFa-}
U(t,x,p)=U(t,x,F(a-))+\lambda_1 a \bigl(F(a)-F(a-)\bigr).
\end{equation}
Also note from (\ref{eqc-pchar1}) that
\[
a\bigl(F(a)-F(a-)\bigr)=U(t,x,F(a))-U(t,x,F(a-)).
\]
Plugging this back into (\ref{UFaFa-}), we obtain
%
\begin{equation}\label{Ulambda1,2}
U(t,x,p)=\lambda_1U(t,x,F(a))+\lambda_2U(t,x,F(a-)).
\end{equation}
It then follows from (\ref{plambda1,2}) and (\ref{Ulambda1,2}) that
%
\begin{eqnarray}\label{eqw=tw1}\quad
&&
pq- U(t,x,p)\nonumber\\
&&\qquad= \lambda_1[F(a)q-U(t,x,F(a))]+\lambda
_2[F(a-)q-U(t,x,F(a-))]\\
&&\qquad\le \max\{F(a)q-U(t,x,F(a)),F(a-)q-U(t,x,F(a-))\}.
\nonumber
\end{eqnarray}
Choose a sequence $a_n\in[a/2,a)$ such that $a_n\to a$ from the left as
$n\to\infty$. Thanks to Proposition \ref{c-convex}, $p\mapsto U(t,x,p)$
is continuous on $[0,1]$. We can therefore select a subsequence of
$a_n$ (without relabelling) such that, for any $n\in\mathbb{N}$,
\[
F(a-)-F(a_n)<\frac{1}{n} \quad\mbox{and}\quad
U(t,x,F(a_n))-U(t,x,F(a-))<\frac{1}{n}.
\]
It follows that, for any $n\in\mathbb{N}$,
\[
F(a-)q-U(t,x,F(a-)) < F(a_n)q-U(t,x,F(a_n))+\frac{1+q}{n},
\]
which yields
%
\begin{eqnarray}\label{eqw=tw2}
&&F(a-)q-U(t,x,F(a-)) \nonumber\\
&&\qquad\le \limsup_{n\to\infty}\biggl\{
F(a_n)q-U(t,x,F(a_n))+\frac{1+q}{n}\biggr\}\\
&&\qquad\le \sup_{n\in\mathbb{N}}F(a_n)q-U(t,x,F(a_n)).\nonumber
\end{eqnarray}
Combining (\ref{eqw=tw1}) and (\ref{eqw=tw2}), we obtain
\[
pq- U(t,x,p)\le\sup_{\delta\in[a/2,a]}F(\delta)q-U(t,x,F(\delta
))\le
\sup_{\delta\ge0}F(\delta)q-U(t,x,F(\delta)).
\]
This implies
\[
w(t,x,q)=\sup_{p\in[0,1]}\{pq-U(t,x,p)\}\le\sup_{a\ge0}\{
F(a)q-U(t,x,F(a))\}.
\]
Since $F(a)\in[0,1]$ for all $a\ge0$, the opposite inequality is
trivial. We therefore conclude
%
\begin{equation}\label{prop1-1}
w(t,x,q)=\sup_{p\in[0,1]}\{pq-U(t,x,p)\}=\sup_{a\ge0}\{
F(a)q-U(t,x,F(a))\}.\hspace*{-32pt}
\end{equation}
Now, thanks to (\ref{equrep1}), we have
%
\begin{eqnarray}\label{prop1-2}
F(a)q-U(t,x,F(a))&=& F(a)q-\mathbb{E}[Z^{t,x,1}(T)g(X^{t,x}(T))
1_{\bar{A}_a}]\nonumber\\[-8pt]\\[-8pt]
&=&
\mathbb{E}\bigl[\bigl(q-Z^{t,x,1}(T)g(X^{t,x}(T))\bigr)1_{\bar{A}_a}\bigr].\nonumber
\end{eqnarray}
It follows from (\ref{prop1-1}), (\ref{prop1-2}) and Lemma \ref
{lemmax=tw} that
\[
w(t,x,q)=\max_{a\ge0}\mathbb{E}\bigl[\bigl(q-Z^{t,x,1}(T)g(X^{t,x}(T))\bigr)
1_{\bar{A}_a}\bigr]=\widetilde{w}(t,x,q).
\]
\upqed\end{pf}

\begin{rem}
Since $w=\widetilde{w}$, we immediately have the following result from
Proposition \ref{propviscositytw}: $w$ is a continuous viscosity
solution to (\ref{PDEtw}) on $(0,T)\times(0,\infty)^d\times
(0,\infty)$
with the boundary condition (\ref{boundarytw}).
\end{rem}

\subsection{\texorpdfstring{Viscosity supersolution property of $U$.}{Viscosity supersolution property of $U$}} Let us extend the
domain of the map $q\mapsto\widetilde{w}_\varepsilon(t,x,q)$ from
$(0,\infty)$ to the entire real line $\mathbb{R}$ by setting
$\widetilde
{w}_\varepsilon(t,x,0)=0$ and $\widetilde{w}_\varepsilon
(t,x,q)=\infty$
for $q<0$. In this subsection, we consider the Legendre transform of
$\widetilde{w}_\varepsilon$ with respect to the $q$ variable
\[
U_\varepsilon(t,x,p):=\sup_{q\in\mathbb{R}}\{pq-\widetilde
{w}_\varepsilon(t,x,q)\}=\sup_{q\ge0}\{pq-\widetilde{w}_\varepsilon
(t,x,q)\}.
\]
We will first show that $U_\varepsilon$ is a classical solution to a
nonlinear PDE. Then we will relate $U_\varepsilon$ to $U$ and derive
the viscosity supersolution property of $U$.
%
\begin{prop}
Under Assumption \ref{aslocLips}, we have that $U_\varepsilon\in\break
\mathcal{C}^{1,2,2}((0,T)\times(0,\infty)^d\times(0,1))$ and satisfy
the equation
%
\begin{eqnarray}\label{PDEUe}
0&=&\partial_t U_\varepsilon+\frac{1}{2}\operatorname{Tr}[\sigma
\sigma'
D_{xx}U_\varepsilon]\nonumber\\
&&{}+\inf_{a\in\mathbb{R}^d}
\biggl((D_{xp}U_\varepsilon
)'\sigma a
+\frac{1}{2}|a|^2D_{pp}U_\varepsilon-\theta'aD_pU_\varepsilon\biggr)\\
&&{}+\inf_{b\in\mathbb{R}^d}\biggl(\frac{1}{2}|b|^2D_{pp}U_\varepsilon
-\varepsilon D_pU_\varepsilon\mathbf{1}'b\biggr),\nonumber
\end{eqnarray}
where ${\mathbf1}:=(1,\ldots,1)'\in\mathbb{R}^d$, with the boundary condition
%
\begin{equation}\label{boundaryUe}
U_\varepsilon(T,x,p)=pg(x).
\end{equation}
Moreover, $U_\varepsilon(t,x,p)$ is strictly convex in the $p$ variable
for $p\in(0,1)$, with
%
\begin{equation}\label{DpUe}
\lim_{p\downarrow0}D_pU_\varepsilon(t,x,p)=0\quad \mbox{and}\quad \lim
_{p\uparrow1}D_pU_\varepsilon(t,x,p)=\infty.
\end{equation}
\end{prop}
\begin{pf}
Since from Proposition \ref{propstrictconvex} the function $q\mapsto
D_q\widetilde{w}_\varepsilon(t,x,q)$ is strictly increasing on
$(0,\infty)$ with
\[
\lim_{q\downarrow0}D_q\widetilde{w}_\varepsilon(t,x,q)=0
\quad\mbox
{and}\quad
\lim_{q\to\infty}D_q\widetilde{w}_\varepsilon(t,x,q)=1,
\]
its inverse function $p\mapsto H(t,x,p)$ is well defined on $(0,1)$.
Moreover, considering that $\widetilde{w}_\varepsilon(t,x,q)$ is smooth
on $(0,T)\times(0,\infty)^d\times(0,\infty)$, $U_\varepsilon
(t,x,p)$ is
smooth on $(0,T)\times(0,\infty)^d\times(0,1)$ and can be expressed as
%
\begin{eqnarray}\label{eqUepsilon=sup}
U_\varepsilon(t,x,p)&=&\sup_{q\ge0}\{pq-\widetilde{w}_\varepsilon
(t,x,q)\}\nonumber\\[-8pt]\\[-8pt]
&=&pH(t,x,p)-\widetilde{w}_\varepsilon(t,x,H(t,x,p));\nonumber
\end{eqnarray}
see, for example,~\cite{Roc97}.
By direct calculations, we have
%
\begin{eqnarray}\label{derivatives}
D_p U_\varepsilon(t,x,p) &=& H(t,x,p),\nonumber\\[-1pt]
D_{pp} U_\varepsilon(t,x,p) &=& D_p H(t,x,p) = \frac
{1}{D_{qq}\widetilde
{w}_\varepsilon(t,x,H(t,x,p))},\nonumber\\[-1pt]
D_{x} U_\varepsilon(t,x,p) &=& -D_x \widetilde{w}_\varepsilon
(t,x,H(t,x,p)),\nonumber\\[-1pt]
D_{xx} U_\varepsilon(t,x,p) &=& -D_{xx} \widetilde{w}_\varepsilon
(t,x,H(t,x,p))\\[-1pt]
&&{}+\frac{1}{D_{pp} U_\varepsilon
(t,x,p)}(D_{px}U_\varepsilon
)(D_{px}U_\varepsilon)',\nonumber\\[-1pt]
D_{px} U_\varepsilon(t,x,p) &=& -D_{qx} \widetilde{w}_\varepsilon
(t,x,H(t,x,p)) D_{pp} U_\varepsilon(t,x,p),\nonumber\\[-1pt]
\partial_{t} U_\varepsilon(t,x,p) &=& -\partial_{t} \widetilde
{w}_\varepsilon(t,x,H(t,x,p)).\nonumber
\end{eqnarray}
In particular, we see that $U_\varepsilon(t,x,p)$ is strictly convex in
$p$ for $p\in(0,1)$ and satisfies (\ref{DpUe}). Now by setting
$q:=H(t,x,p)$, we deduce from (\ref{PDEtwe}) that
%
\begin{eqnarray}\label{twetoUe}\quad
0 &=& -\partial_t
\widetilde{w}_\varepsilon-\frac{1}{2}\operatorname{Tr}[\sigma
\sigma'
D_{xx}\widetilde{w}_\varepsilon]\nonumber\\[-1pt]
&&{}-\frac{1}{2}(|\theta|^2+\varepsilon
^2)q^2D_{qq}\widetilde{w}_\varepsilon-q\operatorname{Tr}[\sigma\theta
D_{xq}\widetilde
{w}_\varepsilon]\nonumber\\[-1pt]
&=&
\partial_t U_\varepsilon+\frac{1}{2}\operatorname{Tr}[\sigma\sigma'
D_{xx}U_\varepsilon]-\frac{1}{2D_{pp}U_\varepsilon}\operatorname
{Tr}[\sigma\sigma'
(D_{px}U_\varepsilon)(D_{px}U_\varepsilon)']\nonumber\\[-1pt]
&&{}-\frac{1}{2}(|\theta
|^2+\varepsilon^2)\frac{(D_p U_\varepsilon)^2}{D_{pp}U_\varepsilon
}
+\frac{D_p U_\varepsilon}{D_{pp}U_\varepsilon
}\operatorname{Tr}[\sigma
\theta D_{px}U_\varepsilon]\nonumber\\[-1pt]
&=&
\partial_t U_\varepsilon+\frac{1}{2}\operatorname{Tr}[\sigma\sigma'
D_{xx}U_\varepsilon]\nonumber\\[-9pt]\\[-9pt]
&&{}+\biggl((D_{xp}U_\varepsilon)'\sigma a^*+\frac
{1}{2}|a^*|^2D_{pp}U_\varepsilon-\theta'a^*D_p U_\varepsilon
\biggr)\nonumber\\[-1pt]
&&{}+\biggl(\frac{1}{2}|b^*|^2D_{pp}U_\varepsilon-\varepsilon
D_pU_\varepsilon{\mathbf1}'b^*\biggr)\nonumber\\[-1pt]
&=&
\partial_t U_\varepsilon+\frac{1}{2}\operatorname{Tr}[\sigma\sigma'
D_{xx}U_\varepsilon]\nonumber\\[-1pt]
&&{}+\inf_{a\in\mathbb{R}^d}
\biggl((D_{xp}U_\varepsilon
)'\sigma a+\frac{1}{2}|a|^2D_{pp}U_\varepsilon-\theta'a D_p
U_\varepsilon\biggr)\nonumber\\[-1pt]
&&{}+\inf_{b\in\mathbb{R}^d}\biggl(\frac{1}{2}|b|^2D_{pp}U_\varepsilon
-\varepsilon D_pU_\varepsilon{\mathbf1}'b\biggr),\nonumber
\end{eqnarray}
where the minimizers $a^*$ and $b^*$ are defined by
\begin{eqnarray*}
a^*(t,x,p)&:=& \frac{D_p U_\varepsilon(t,x,p)}{D_{pp} U_\varepsilon
(t,x,p)}\theta(x)-\frac{1}{D_{pp} U(t,x,p)}\sigma'(x)D_{px}
U_\varepsilon(t,x,p),\\[-1pt]
b^*(t,x,p)&:=& \varepsilon\frac{D_pU_\varepsilon
(t,x,p)}{D_{pp}U_\varepsilon(t,x,p)}{\mathbf1}.\vadjust{\goodbreak}
\end{eqnarray*}
Finally, observe that for any $p\in(0,1)$, the maximum of
$pq-(q-g(x))^+$ is attained at $q=g(x)$. Therefore, by (\ref{boundarytwe})
\begin{eqnarray*}
U_\varepsilon(T,x,p)&=&\sup_{q\ge0}\{pq-\widetilde{w}_\varepsilon
(T,x,p)\}= \sup_{q\ge0}\bigl\{pq-\bigl(q-g(x)\bigr)^+\bigr\}\\
&=&pg(x).
\end{eqnarray*}
\upqed\end{pf}

Now we intend to use the stability of viscosity solutions to derive the
supersolution property of $U$. We first have the following observation.
%
\begin{lem}\label{lemU=liminfUe}
For any $(t,x,p)\in[0,T]\times(0,\infty)^d\times\mathbb{R}$, we have
\[
\liminf_{(\varepsilon,\tilde{t},\tilde{x},\tilde{p})\to
(0,t,x,p)}U_\varepsilon(\tilde{t},\tilde{x},\tilde{p})=U(t,x,p).
\]
\end{lem}
\begin{pf}
As a consequence of Lemma \ref{lemtwtwe}(ii), $\widetilde
{w}_\varepsilon(t,x,q)$ is continuous at $(\varepsilon,t,x,q)\in
[0,\infty)\times[0,T]\times(0,\infty)^d\times(0,\infty)$. This implies
that $U_\varepsilon(t,x,p)=\sup_{q\ge0}\{pq-\widetilde
{w}_\varepsilon
(t,x,q)\}$ is lower semicontinuous at $(\varepsilon,t,x,p)\in
[0,\infty
)\times[0,T]\times(0,\infty)^d\times\mathbb{R}$. It follows that
\begin{eqnarray*}
\liminf_{(\varepsilon,\tilde{t},\tilde{x},\tilde{p})\to
(0,t,x,p)}U_\varepsilon(\tilde{t},\tilde{x},\tilde{p})&=&\sup_{q\ge
0}\{
pq-\widetilde{w}(t,x,q)\}=\sup_{q\ge0}\{pq-{w}(t,x,q)\}\\
&=&U(t,x,p),
\end{eqnarray*}
where the second equality follows from Proposition \ref{propw=tw}.
\end{pf}

Before we state the supersolution property for $U$, let us first
introduce some notation. For any $(x,\beta,\gamma,\lambda)\in
(0,\infty
)^d\times\mathbb{R}\times\mathbb{R}\times\mathbb{R}^d$, define
\[
G(x,\beta,\gamma,\lambda):=\inf_{a\in\mathbb{R}^d}\biggl(\lambda
'\sigma
(x)a+\frac{1}{2}|a|^2\gamma-\beta\theta(x)'a\biggr).
\]
We also consider the lower semicontinuous envelope of $G$
\[
G_*(x,\beta,\gamma,\lambda):=\liminf_{(\tilde{x},\tilde{\beta
},\tilde
{\gamma},\tilde{\lambda})\to(x,\beta,\gamma,\lambda)} G(\tilde
{x},\tilde
{\beta},\tilde{\gamma},\tilde{\lambda}).
\]
Observe that, by definition,
%
\begin{equation}\label{Hlower*}
G_*(x,\beta,\gamma,\lambda)=\cases{
G(x,\beta,\gamma,\lambda), &\quad if $\gamma
>0$;\cr
-\infty, &\quad if $\gamma\le0$.}
\end{equation}

\begin{prop}\label{propviscosityU}
Under Assumption \ref{aslocLips}, $U$ is a lower semicontinuous
viscosity supersolution to the equation
%
\begin{equation}\label{PDEU}
0 \ge\partial_t U+\tfrac{1}{2}\operatorname{Tr}[\sigma\sigma'
D_{xx}U]+G_*(x,D_{p}U,D_{pp}U,D_{xp}U)
\end{equation}
for $(t,x,p)\in(0,T)\times(0,\infty)^d\times(0,1)$, with the
boundary condition
%
\begin{equation}\label{boundaryU}
U(T,x,p) = pg(x).
\end{equation}
\end{prop}
\begin{pf}
Note that the lower semicontinuity of $U$ is a consequence of
Lem\-ma~\ref
{lemU=liminfUe}, and the boundary condition (\ref{boundaryU}) comes
from the fact that $w=\widetilde{w}$ and the definition of $\widetilde
{w}$ as the following calculation demonstrates:
\begin{eqnarray*}
U(T,x,p)&=&\sup_{q\ge0}\{pq-w(T,x,p)\}=\sup_{q\ge0}\{pq-\widetilde
{w}(T,x,p)\}\\
&=&\sup_{q\ge0}\bigl\{pq-\bigl(q-g(x)\bigr)^+\bigr\}=pg(x).
\end{eqnarray*}

Let us now turn to the PDE characterization inside the domain of $U$.
Set $\bar{x}:=(t,x,p)$. Let $\varphi$ be a smooth function such that
$U-\varphi$ attains a local minimum at $\bar{x}_0=(t_0,x_0,p_0)\in
(0,T)\times(0,\infty)^d\times(0,1)$ and $U(\bar{x}_0)=\varphi(\bar
{x}_0)$. Note from (\ref{Hlower*}) that as $D_{pp}\varphi(\bar
{x}_0)\le0$, we must have $G_*(x_0,D_{p}\varphi,D_{pp}\varphi
,D_{xp}\varphi)=-\infty$. Thus, the viscosity supersolution property
(\ref{PDEU}) is trivially satisfied. We therefore assume in the
following that $D_{pp}\varphi(\bar{x}_0)> 0$.

Let $F_\varepsilon(\bar{x},\partial_tU_\varepsilon(\bar
{x}),D_pU_\varepsilon(\bar{x}),D_{pp}U_\varepsilon(\bar
{x}),D_{xp}U_\varepsilon(\bar{x}),D_{xx}U_\varepsilon(\bar
{x}))$
denote the\break right-hand side of (\ref{PDEUe}). Observe from the
calculation in (\ref{twetoUe}) that as $\gamma>0$,
\begin{eqnarray*}
F_\varepsilon(\bar{x},\alpha,\beta,\gamma,\lambda,A)&=& \alpha
+\frac
{1}{2}\operatorname{Tr}[\sigma(x)\sigma(x)'A]-\frac{1}{2\gamma
}\operatorname{Tr}[\sigma(x)\sigma
(x)'\lambda\lambda']\\
&&{}-\frac{\beta^2}{2\gamma}\bigl(|\theta
(x)|^2+\varepsilon
^2\bigr)+\frac{\beta}{\gamma}\operatorname{Tr}[\sigma(x)\theta
(x)\lambda].
\end{eqnarray*}
This shows that $F_\varepsilon$ is continuous at every $(\varepsilon
,\bar{x},\alpha,\beta,\gamma,\lambda,A)$ as long as $\gamma>0$. It
follows that for any $z=(\bar{x},\alpha,\beta,\gamma,\lambda,A)$ with
$\gamma>0$, we have
%
\begin{eqnarray}\label{F0}
F_*(z):\!&=&\liminf_{(\varepsilon,z')\to(0,z)}F_\varepsilon(z')=F_0(z)\nonumber\\
&=&
\alpha+\frac{1}{2}\operatorname{Tr}[\sigma(x)\sigma(x)'A]\\
&&{}+\inf
_{a\in\mathbb{R}^d}
\biggl(\lambda'\sigma(x)a+\frac{1}{2}|a|^2\gamma-\theta(x)'a\beta\biggr).\nonumber
\end{eqnarray}
Since we have $U(\bar{x})=\liminf_{(\varepsilon,\bar{x}')\to
(0,\bar
{x})} U_\varepsilon(\bar{x}')$ from Lemma \ref{lemU=liminfUe}, we may
use the same argument in~\cite{Touzi-note-Pisa}, Proposition 2.3, and
obtain that
\[
F_*(\bar{x}_0,\partial_t\varphi(\bar{x}_0),D_p\varphi(\bar
{x}_0),D_{pp}\varphi(\bar{x}_0),D_{xp}\varphi(\bar
{x}_0),D_{xx}\varphi
(\bar{x}_0)) \le0.
\]
Considering that $D_{pp}\varphi(\bar{x}_0)>0$, we see from (\ref{F0})
and (\ref{Hlower*}) that this is the desired supersolution property.
\end{pf}

A few remarks are in order:
%
\begin{rem}
Results similar to Proposition \ref{propviscosityU} were proved by
\cite{BET08}, with stronger assumptions [such as the existence of an
equivalent martingale measure and the existence of a unique strong
solution to (\ref{eqstk})], using the stochastic target formulation.
Here, we first observe that the Legendre transform of $U$ is equal to
$\widetilde{w}$ and that $\widetilde{w}$ can be approximated by
$\widetilde{w}_\varepsilon$, which is a classical solution to a linear
PDE and is strictly convex in $q$; then, we apply the Legendre duality
argument, as carried out in~\cite{KLS}, to show that $U_\varepsilon$,
the Legendre transform of $\widetilde{w}_\varepsilon$, is a classical
solution to a nonlinear PDE. Finally, the stability of viscosity
solutions leads to the viscosity supersolution property of $U$.
\end{rem}
%
\begin{rem}
Instead of relying on the Legendre duality we could directly apply the
dynamic programming principle of~\cite{Haussmann-Lepeltier} for weak
solutions to the formulation in Section \ref{secstoc-cont}. The
problem with this approach is that it requires some growth conditions
on the coefficients of (\ref{eqstk}), which would rule out the
possibility of arbitrage, the thing we are interested in and want to
keep in the scope of our discussion.
\end{rem}
%
\begin{rem}
Under our assumptions, the solution of (\ref{PDEU}) may not be unique
as pointed out below:
\begin{longlist}
\item Let us consider the PDE satisfied by the superhedging price
$U(t,x,1)$,
%
\begin{eqnarray}
\label{eqpde1}
&0 = v_t + \tfrac1 2 \operatorname{Tr}(\sigma\sigma' D_x^2 v)\qquad
\mbox{on }
(0,T)\times(0,\infty)^d,&
\\
%
\label{eqtc1}
&v(T-,x) = g(x)\qquad \mbox{on } (0,\infty)^d.&
\end{eqnarray}
Unless additional boundary conditions are specified, this PDE may have
multiple solutions.
The role of additional boundary conditions in identifying $(t,x) \to
U(t,x,1)$ as the unique solution of the above Cauchy problem is
discussed in Section 4 of~\cite{bkx10}. Also see~\cite{MR2454711} for a
similar discussion on boundary conditions for degenerate parabolic
problems on bounded domains.\looseness=1

Even when additional boundary conditions are specified, the growth
of~$\sigma$ might lead to the loss of uniqueness; see, for example,
\cite{BX09} and Theorem~4.8 of~\cite{bkx10} which give necessary and
sufficient conditions on the uniqueness of Cauchy problems in one and
two-dimensional settings in terms of the growth rate of its
coefficients.
We also note that~\cite{FK10OA} develops necessary and sufficient
conditions for uniqueness, in terms of the attainability of the
boundary of the positive orthant by an auxiliary diffusion (or, more
generally, an auxiliary It\^{o}) process.\looseness=0

\item Let $\Delta U(t,x,1)$ be the difference of two solutions of
(\ref{eqpde1})--(\ref{eqtc1}). Then both $U(t,x,p)$ and
$U(t,x,p)+\Delta U(t,x,1)$ are solutions of (\ref{PDEU}) (along with
its boundary conditions). As a result, whenever (\ref{eqpde1}) and
(\ref{eqtc1}) have multiple solutions, so does the PDE (\ref{PDEU})
for the value function~$U$.
\end{longlist}
\end{rem}

\subsection{\texorpdfstring{Characterizing the value function $U$.}{Characterizing the value function $U$}} We intend to
characterize $U_\varepsilon$ as the smallest solution among a
particular class of functions, as specified below in Proposition \ref
{propcharacUe}. Then, considering that $\liminf_{(\varepsilon,\tilde
{t},\tilde{x},\tilde{p})\to(0,t,x,p)}U_\varepsilon(\tilde
{t},\allowbreak \tilde{x},\tilde{p})=U(t,x,p)$ from Lemma \ref{lemU=liminfUe}, this gives a
characterization for $U$. In determining $U$ numerically, one could use
$U_{\eps}$ as a proxy for $U$ for small enough $\eps$. Additionally, we
will characterize $U$ as the smallest nonnegative supersolution of
(\ref{PDEU}) in Proposition \ref{propcharacU}.
%
\begin{prop}\label{propcharacUe}
Suppose that Assumption \ref{aslocLips} holds. Let $u\dvtx[0,T]\times
(0,\infty)^d\times[0,1]\mapsto[0,\infty)$ be of class $\mathcal
{C}^{1,2,2}((0,T)\times(0,\infty)^d\times(0,1))$ such that
$u(t,x,0)=0$, and $u(t,x,p)$ is strictly convex in $p$ for $p\in(0,1)$ with
%
\begin{equation}\label{Dpu}
\lim_{p\downarrow0}D_pu(t,x,p)=0 \quad\mbox{and} \quad\lim_{p\uparrow
1}D_pu(t,x,p)=\infty.
\end{equation}
If $u$ satisfies the partial differential inequality
%
\begin{eqnarray}\label{PDEu}
0 &\ge&\partial_t u+\frac{1}{2}\operatorname{Tr}[\sigma\sigma'
D_{xx}u]+\inf_{a\in
\mathbb{R}^d}\biggl((D_{xp}u)'\sigma a
+\frac{1}{2}|a|^2D_{pp}u-\theta'aD_p u\biggr)\nonumber\hspace*{-35pt}\\[-8pt]\\[-8pt]
&&{}+\inf_{b\in\mathbb
{R}^d}\biggl(\frac{1}{2}|b|^2D_{pp}u-\varepsilon D_p u{\mathbf1}'b\biggr),\nonumber\hspace*{-35pt}
\end{eqnarray}
where ${\mathbf1}:=(1,\ldots,1)'\in\mathbb{R}^d$, with the boundary condition
%
\begin{equation}\label{boundaryu}
u(T,x,p)=pg(x),
\end{equation}
then $u\ge U_\varepsilon$.
\end{prop}
\begin{pf}
Let us extend the domain of the map $p\mapsto u(t,x,p)$ from $[0,1]$ to
the entire real line $\mathbb{R}$ by setting $u(t,x,p)=0$ for $p<0$ and
$u(t,x,p)=\infty$ for $p>1$. Then we can define the Legendre transform
of $u$ with respect to the $p$ variable
%
\begin{eqnarray}\label{wu>0}
w^u(t,x,q):\!&=&\sup_{p\in\mathbb{R}}\{pq-u(t,x,p)\}\nonumber\\[-8pt]\\[-8pt]
&=&\sup_{p\in[0,1]}\{pq-u(t,x,p)\}\ge0\qquad \mbox{for } q\ge0,\nonumber
\end{eqnarray}
where the positivity comes from the condition $u(t,x,0)=0$. First,
observe that since $u$ is nonnegative, we must have
%
\begin{equation}\label{wu<q}
w^u(t,x,q)\le\sup_{p\in[0,1]}pq = q\qquad \mbox{for any } q\ge0.
\end{equation}
Next, we derive the boundary condition of $w^u$ from (\ref{boundaryu}) as
%
\begin{eqnarray}\label{boundarywu}
w^u(T,x,q)&=&\sup_{p\in[0,1]}\{pq-u(T,x,p)\}=\sup_{p\in[0,1]}\{
pq-pg(x)\}\nonumber\\[-8pt]\\[-8pt]
&=&\bigl(q-g(x)\bigr)^+.\nonumber
\end{eqnarray}
Now, since $u(t,x,p)$ is strictly convex in $p$ for $p\in(0,1)$ and
satisfies (\ref{Dpu}), we can express $w^u$ as
\[
w^u(t,x,q)=J(t,x,q)q-u(t,x,J(t,x,q))\qquad \mbox{for } q\in(0,\infty),
\]
where $q\mapsto J(\cdot,q)$ is the inverse function of $p\mapsto D_p
u(\cdot,p)$. We can therefore compute the derivatives of $w^u(t,x,q)$
in terms of those of $u(t,x,J(t,x,q))$, as carried out in (\ref
{derivatives}). We can then perform the same calculation in (\ref
{twetoUe}) (but going backward), and deduce from (\ref{PDEu}) that for any
$(t,x,q)\in(0,T)\times(0,\infty)^d\times(0,\infty)$,
%
\begin{eqnarray}\label{PDEwu}
0 &\le&\partial_t w^u+\tfrac{1}{2}\operatorname{Tr}[\sigma\sigma'
D_{xx}w^u]+\tfrac
{1}{2}(|\theta|^2+\varepsilon^2)q^2D_{qq}w^u\nonumber\\[-8pt]\\[-8pt]
&&{}+q\operatorname{Tr}[\sigma\theta D_{xq}w^u].\nonumber
\end{eqnarray}
Define the process $Y(s):=Z^{t,x,1}(s)Q^{t,x,q}_\varepsilon(s)$ for
$s\in[t,T]$. Observing that
\[
Y(s)=
q\exp\bigl\{-\tfrac{1}{2}\varepsilon^2(s-t)+\varepsilon\bigl(B(s)-B(t)\bigr)\bigr\},
\]
we conclude that $Y(\cdot)$ is a martingale with $\bar{\mathbb
{E}}[Y(s)]=q$ and $\operatorname{Var}(Y(s))=q^2(e^{\varepsilon
^2(s-t)}-1)$ for
all $s\in[t,T]$, and satisfies the following SDE:
\[
dY(s) = \varepsilon Y(s) \,dB(s) \qquad\mbox{for } s\in[t,T]\quad
\mbox{and}\quad Y(t)=q.
\]
Thanks to the Burkholder--Davis--Gundy inequality, there exists a
constant $C>0$ such that
%
\begin{eqnarray}\label{Ybdd}
\bar{\mathbb{E}}\Bigl[{\max_{t\le s\le T}}|Y(s)|^2\Bigr]&\le& C\bar
{\mathbb
{E}}\biggl[\int_{t}^T\varepsilon^2Y^2(s)\,ds\biggr]
\nonumber\\[-8pt]\\[-8pt]
&=&C\varepsilon^2\int_t^T q^2\bigl(e^{\varepsilon^2(s-t)}-1\bigr)+q^2 \,ds
<\infty.\nonumber
\end{eqnarray}
For each $n\in\mathbb{N}$, define the stopping time
\[
\tau_n:=\inf\{
s\ge
t\dvtx|X^{t,x}(s)|>n \mbox{ or } |Q^{t,x,q}_\varepsilon(s)|>n\}.
\]
By applying the product rule to the process $Z^{t,x,1}(\cdot)w^u(\cdot
,X^{t,x}(\cdot),Q^{t,x,q}_\varepsilon(\cdot))$ and using (\ref
{PDEwu}), we get
%
\begin{eqnarray}\label{wu<E}
&&w^u(t,x,q)\nonumber\\
&&\qquad\le\bar{\mathbb{E}}\bigl[Z^{t,x,1}(T\wedge\tau_n)w^u\bigl(T\wedge
\tau
_n,X^{t,x}(T\wedge\tau_n),Q^{t,x,q}_\varepsilon(T\wedge\tau_n)\bigr)\bigr]\\
&&\eqntext{\mbox
{for } n\in\mathbb{N}.}
\end{eqnarray}
Now, observe from (\ref{wu<q}) that
$Z^{t,x,1}(s)w^u(s,X^{t,x}(s),Q^{t,x,q}_\varepsilon(s))\le Y(s)$ for
any $s\in[t,T]$. Then from (\ref{Ybdd}), we may apply the dominated
convergence theorem to (\ref{wu<E}) and obtain
\begin{eqnarray*}
w^u(t,x,q)&\le&\bar{\mathbb
{E}}[Z^{t,x,1}(T)w^u(T,X^{t,x}(T),Q^{t,x,q}_\varepsilon(T))]\\
&=& \bar{\mathbb{E}}\bigl[Z^{t,x,1}(T)\bigl(Q^{t,x,q}_\varepsilon
(T)-g(X^{t,x}(T))\bigr)^+\bigr] = \widetilde{w}_\varepsilon(t,x,q),
\end{eqnarray*}
where the first equality is due to (\ref{boundarywu}). It follows that
\begin{eqnarray*}
u(t,x,p)&=&\sup_{q\ge0}\{pq-w^u(t,x,q)\}\ge\sup_{q\ge0}\{
pq-\widetilde
{w}_\varepsilon(t,x,q)\}\\
&=&U_\varepsilon(t,x,p).
\end{eqnarray*}
\upqed\end{pf}
%
\begin{prop}\label{propcharacU}
Suppose Assumption \ref{aslocLips} holds. Let $u\dvtx[0,T]\times\break(0,\infty
)^d\times[0,1]\mapsto[0,\infty)$ be such that $u(t,x,0)=0$, $u(t,x,p)$
is convex in $p$, and the Legendre transform of $u$ with respect to the
$p$ variable, as defined in the proof of Proposition \ref
{propcharacUe}, is continuous on $[0,T]\times(0,\infty)^d\times
(0,\infty)$. If $u$
is a lower semicontinuous viscosity supersolution to (\ref{PDEU}) on
$(0,T)\times(0,\infty)^d\times(0,1)$ with the boundary condition
$(\ref{boundaryU})$, then $u\ge U$.
\end{prop}
\begin{pf}
Let us denote by $w^u$ the Legendre transform of $u$ with respect to
$p$. By the same argument in the proof of Proposition \ref
{propcharacUe}, we can show that (\ref{wu>0}), (\ref{wu<q}) and
(\ref{boundarywu}) are true. Moreover, as demonstrated in
\cite{BET08}, Section 4, by
using the supersolution property of $u$ we may show that $w^u$ is an
upper semicontinuous viscosity subsolution on $(0,T)\times(0,\infty
)^d\times(0,\infty)$ to the equation
%
\begin{equation}\label{PDEwu2}\quad
\partial_t w^u + \tfrac{1}{2}\operatorname{Tr}(\sigma\sigma'D^2_x
w^u) + \tfrac
{1}{2}|\theta|^2 q^2 D^2_q w^u + q \operatorname{Tr}(\sigma\theta
D_{xq}w^u) = 0.
\end{equation}

Let $\rho(t,x,q)$ be a nonnegative $\mathcal{C}^\infty$ function
supported in $\{(t,x,q)\dvtx t\in[0,1],|(x,q)|\le1\}$ with unit mass.
Without loss of generality, set $w^u(t,x,q)=0$ for $(t,x,q)\in\mathbb
{R}^{d+2}\cap([0,T]\times(0,\infty)^d\times(0,\infty))^c$.
Then for any $(t,x,q)\in\mathbb{R}^{d+2}$, define
\[
w^u_\delta(t,x,q):=\rho^\delta\mathop{*}w^u \qquad\mbox{where } \rho
^\delta
(t,x,q):=\frac{1}{\delta^{d+2}}\rho\biggl(\frac{t}{\delta^2},\frac
{x}{\delta},\frac{q}{\delta}\biggr).
\]
By definition, $w^u_\delta$ is $\mathcal{C}^\infty$. Moreover, it can
be shown that $w^u_\delta$ is a subsolution to (\ref{PDEwu2}) on
$(0,T)\times(0,\infty)^d\times(0,\infty)$; see, for example,
(3.23) and (3.24) in~\cite{FTW10}, Section 3.3.2, and
\cite{BJ02}, Lemma 2.7.
Set $\bar{x}=(t,x,q)$. By (\ref{wu<q}), we see from the definition of
$w^u_\delta$ that
%
\begin{equation}\label{wud<q+d}
w^u_\delta(\bar{x})= \int_{\mathbb{R}^{d+2}}\rho^\delta
(y)w^u(\bar
{x}-y)\,dy \le(q+\delta)\int_{\mathbb{R}^{d+2}}\rho^\delta(y)\,dy =
q+\delta.\hspace*{-32pt}
\end{equation}
Also, the continuity of $w^u$ implies that $w^u_\delta\to w^u$ for
every $(t,x,q)\in[0,T]\times(0,\infty)^d\times(0,\infty)$. Considering
that $w^u_\delta$ is a classical subsolution to (\ref{PDEwu2}), we~have
%
\begin{eqnarray}\label{wud<E}
&&w^u_\delta(t,x,q)\nonumber\\
&&\qquad\le\mathbb{E}\bigl[Z^{t,x,1}(T\wedge\tau_n)w^u_\delta
\bigl(T\wedge\tau_n,X^{t,x}(T\wedge\tau_n),Q^{t,x,q}(T\wedge\tau
_n)\bigr)\bigr]\nonumber\\[-8pt]\\[-8pt]
&&\eqntext{\mbox
{for } n\in\mathbb{N},}
\end{eqnarray}
where $\tau_n:=\inf\{s\ge t\dvtx |X^{t,x}(s)|>n \mbox{ or }
|Q^{t,x,q}(s)|>n\}$. For each fixed $n\in\mathbb{N}$, thanks to
(\ref{wud<q+d}), we may\vadjust{\goodbreak} apply the dominated convergence theorem as we take
the limit $\delta\to0$ in (\ref{wud<E}). We thus get
%
\begin{eqnarray}\label{wu<En}\quad
&&w^u(t,x,q)\nonumber\\[-8pt]\\[-8pt]
&&\qquad\le\mathbb{E}\bigl[Z^{t,x,1}(T\wedge\tau_n)w^u\bigl(T\wedge\tau
_n,X^{t,x}(T\wedge\tau_n),Q^{t,x,q}(T\wedge\tau_n)\bigr)\bigr].\nonumber
\end{eqnarray}
Now by applying the reverse of Fatou's lemma (see, e.g.,
\cite{prob-with-mart}, page 53) to~(\ref{wu<En}), we have
\begin{eqnarray*}
w^u(t,x,q)&\le&\mathbb{E}\Bigl[Z^{t,x,1}(T)\limsup_{n\to\infty
}w^u\bigl(T\wedge
\tau_n,X^{t,x}(T\wedge\tau_n),Q^{t,x,q}(T\wedge\tau_n)\bigr)\Bigr]\\
&\le&\mathbb{E}[Z^{t,x,1}(T)w^u(T,X^{t,x}(T),Q^{t,x,q}(T))]\\
&\le&\mathbb{E}\bigl[Z^{t,x,1}(T)\bigl(Q^{t,x,q}(T)-g(X^{t,x}(T))\bigr)^+\bigr]\\
&=&w(t,x,q),
\end{eqnarray*}
where the second inequality follows from the upper semicontinuity of
$w^u$, and the third inequality is due to (\ref{boundarywu}). Finally,
we conclude that
\[
u(t,x,p)=\sup_{q\ge0}\{pq-w^u(t,x,q)\}\ge\sup_{q\ge0}\{
pq-w(t,x,q)\}=U(t,x,p),
\]
where the first equality is guaranteed by the convexity and the lower
semicontinuity of $u$.
\end{pf}

One should note that $U_{\eps}$ and $U$ satisfy the assumptions stated
in Propositions \ref{propcharacUe} and \ref{propcharacU},
respectively. Therefore, one can indeed see these results as PDE
characterizations of the functions $U_{\eps}$ and $U$.

In this paper, under the context where equivalent martingale measures
need not exist, we discuss the quantile hedging problem and focus on
the PDE characterization for the minimum amount of initial capital
required for quantile hedging. An interesting problem following this is
the construction of the corresponding quantile hedging portfolio. We
leave this problem open for future research.

\section*{\texorpdfstring{Acknowledgments.}{Acknowledgments}}

We would like to thank Johannes Ruf for his feedback. We also would
like to express our gratitude to the anonymous Associate Editor and
referees whose comments helped us improve our paper significantly.


%

\printaddresses

\end{document}